%% file: Weighted_Paired-Domination_Problem_on_Block_Graphs-2016-02-20.tex
\documentclass[12pt]{article}
\usepackage{color,epsfig,citesort,amsmath,fullpage}
\usepackage{algorithm}
\usepackage{algorithmic}

\algsetup{linenosize = \normalsize}
\def\myendproof{{\hfill \vbox{\hrule\hbox{%
   \vrule height1.3ex\hskip0.8ex\vrule}\hrule }}\par}
\newtheorem{theorem}{Theorem}
\newtheorem{lemma}[theorem]{Lemma}

\newenvironment{proof}{{\vspace{-5pt}\it Proof. }}{\myendproof}

\newcommand{\setof}[1]{\{{#1}\}}
\newcommand{\Xomit}[1]{}
\newcommand{\rgb}[1]{#1}

\title{{\bf A Linear-Time Algorithm for the \\ Weighted Paired-Domination Problem on Block Graphs\thanks{This work is partially supported by the National Science Council under the Grants No. NSC-102-2221-E-019-038-, and NSC-103-2221-E-019-034-.}}}

\author{Ching-Chi Lin\thanks{Department of Computer Science and Engineering,
                             National Taiwan Ocean University,
                             Keelung 20224, Taiwan. Corresponding author.
                             Email: lincc@mail.ntou.edu.tw}
        \and
        Cheng-Yu Hsieh\thanks{Department of Computer Science and
                              Information Engineering,
                              National Taiwan University,
                              Taipei 10617, Taiwan.
                              Email: r01922114@ntu.edu.tw}
}
\date{\today}

\begin{document}
\maketitle

\begin{abstract}
In a graph $G = (V,E)$, a vertex subset $S\subseteq V(G)$ is
said to be a dominating set of $G$ if every vertex not in $S$ is
adjacent to a vertex in $S$. A dominating set $S$ of $G$ is called a paired-dominating set of $G$ if the induced subgraph $G[S]$ contains a perfect matching. 
In this paper, we propose an $O(n+m)$-time
algorithm for the weighted paired-domination problem on block
graphs using dynamic programming, which strengthens the results in [Theoret. Comput. Sci., 410(47--49):5063--5071, 2009] and [J. Comb. Optim., 19(4):457--470, 2010]. Moreover, the algorithm can be completed in $O(n)$ time if the block-cut-vertex structure of $G$ is given.


\bigskip

\noindent \textbf{Keywords:} \rgb{Weighted} paired-domination problem, perfect
matching, block graph, dynamic programming\rgb{.}

\end{abstract}

\newpage

\def\skippt{22.2pt}
\baselineskip \skippt

\section{Introduction}\label{section:intro}
In a graph $G = (V,E)$, a vertex subset $S\subseteq V(G)$ is
said to be a \emph{dominating set} of $G$ if every vertex not in $S$ is
adjacent to a vertex in $S$. Let $G[S]$ denote the subgraph of $G$ induced by a subset $S$ of $V(G)$. A dominating set $S$ of $G$ is called
a \emph{paired-dominating set} if the induced subgraph $G[S]$ contains a
perfect matching. The \emph{paired-domination problem} involves finding a
paired-dominating set $S$ of $G$ such that the cardinality of $S$
is minimized. Suppose that, for each $v \in V(G)$, we have a
weight $w(v)$ specifying the cost for adding $v$ to $S$. The
\emph{weighted paired-domination problem} is to find a paired-dominating set $S$ whose \rgb{$w(S) = \sum_{v \in
S} {w(v)}$} is minimized.

The domination problem has been extensively studied in the area of algorithmic graph theory for \rgb{several} decades; see~\cite{Hedetniemi90,Hedetniemi91,Haynes98,Haynes98-2,Goddard13,Henning09,Chang13} for books and survey papers. 
It has many applications in the real world such as location problems, communication networks, and kernels of games~\cite{Haynes98}. Depending
on the requirements of different types of applications, there are
several variants of the domination problem, such as the
independent domination, connected domination, total domination, and perfect domination problems~\cite{Chang13,Goddard13,Henning09,Yen96}. These
problems have been proved to be NP-complete and have
polynomial-time algorithms on some special classes of graphs.
In particular, Haynes and Slater~\cite{HS98} introduced the
concept of paired-domination motivated by security concerns. In a
museum protection program, beside the requirement that each region
has a guard in it or is in the protection range of some guard, the
guards must be able \rgb{to} back each other up. 
%
%

\rgb{The \emph{paired-domination number} $\gamma_p(G)$ is the minimum cardinality of a paired-dominating set. In~\cite{HS98}, Haynes and Slater showed that the problem of determining whether $\gamma_p(G) \le c$ is NP-complete on general graphs and gave a lower bound of $n/\Delta(G)$ for $\gamma_p(G)$, where $c$ is a positive even integer, $n$ is the number of the vertices in $G$, and $\Delta(G)$ is the maximum degree of $G$.} Recently, many studies have been made for this problem in proving NP-completeness, providing approximation algorithms, and finding polynomial-time algorithms on some special classes of graphs. Here, we only mention some  related results. For more detailed information regarding this problem, please refer to~\cite{Kang13}. Chen~{\em et al.}~\cite{Chen10} demonstrated that the paired-domination problem is also NP-complete on bipartite graphs, chordal graphs, and split graphs. In~\cite{Chen09}, Chen~{\em et al.} proposed an approximation algorithm with ratio $\ln(2\Delta(G))+1$ for general graphs and showed that the problem is APX-complete, i.e., has no PTAS. Panda and Pradhan~\cite{Panda12} strengthened the results in~\cite{Chen10} by showing that the problem is also NP-complete for perfect elimination bipartite graphs.
						
Meanwhile, polynomial-time algorithms have been studied intensively on some special classes of graphs such as tree graphs~\cite{QKCD03}, weighted tree graphs~\cite{Chen09}, inflated tree graphs~\cite{KSC04}, convex bipartite graph~\cite{Hung12,Panda13}, permutation graphs~\cite{CKS09,Lappas09,Lappas13}, strongly chordal graphs~\cite{Chen-Lu-Zeng-09}, interval graphs~\cite{Chen10} and circular-arc graphs~\cite{lin15}. Especially, Chen~{\em et al.}~\cite{Chen10} introduced an $O(m+n)$-time algorithm for block graphs, a proper superfamily of tree graphs. In this paper, we propose an $O(n+m)$-time algorithm for the weighted paired-domination problem on block graphs using dynamic programming, which strengthens the results in~\cite{Chen09,Chen10}. Moreover, the algorithm can be completed in $O(n)$ time if the block-cut-vertex structure of $G$ is given. Notice that the block-cut-vertex structure of a block graph $G$ can be constructed in $O(n+m)$ time by the depth first search algorithm~\cite{Aho74}. 
%

The remainder of this paper is organized as follows. In Section~\ref{section:algo-dynamic}, given the block-cut-vertex structure of a block graph $G$, we employ dynamic programming to present an $O(n)$-time algorithm for finding a minimum-weight paired-dominating set of $G$. In Section~\ref{section:efficient-implementation}, the correctness proof and complexity analysis of the algorithm are provided. Section~\ref{section:conclusion} contains some concluding remarks and future work.

\section{The Proposed Algorithm for Block Graphs}
\label{section:algo-dynamic}
In this section, given a weighted block graph $G$ with the block-cut-vertex structure $G^*$ of $G$, we propose an $O(n)$-time algorithm that determines a minimum-weight paired-dominating set of $G$ using dynamic programming. \rgb{Since a graph $G$ containing isolated vertices has no paired-dominating set, we suppose that $G$ is a connected graph without isolated vertices in the rest of this paper.} First, we introduce some preliminaries for block graphs. 

\vspace{12pt}
\begin{figure}[thb]
\centerline{\input{block_graph}} \caption{$(a)$ A block graph $G$. $(b)$
The corresponding block-cut-vertex graph $G^*$ for the block graph
$G$ in ($a$). In particular, Algorithm 1 considered the blocks of $G$ in turn according to the ordering $B_1, B_2,\ldots, B_8.$} \label{fig:block_graph}
\end{figure}
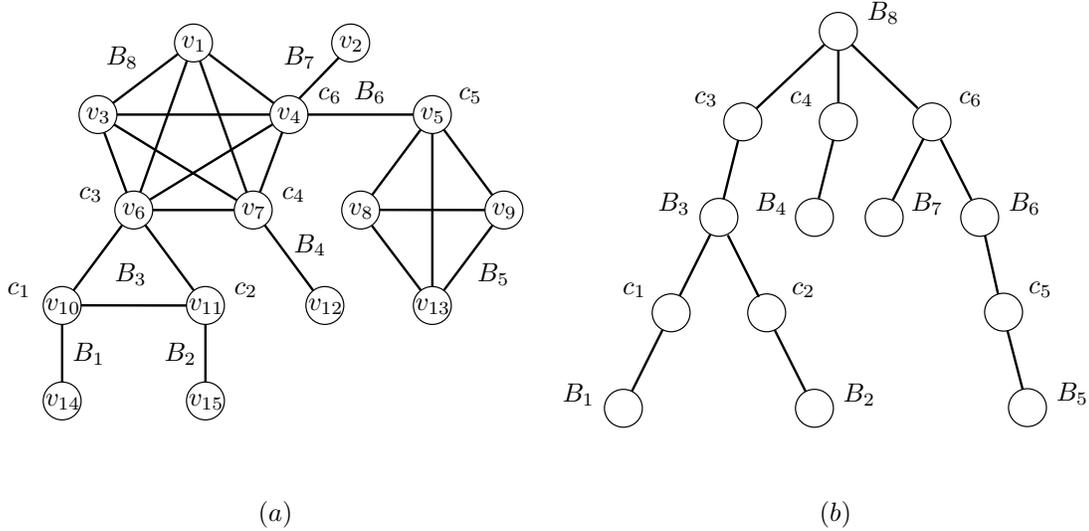

For any connected graph $G$, a vertex $x\in V(G)$ is called a
\emph{cut-vertex} of $G$, if $G-x$ contains more than one
connected component. A \emph{block} is a maximal connected
subgraph without a cut-vertex. A graph $G$ is called a \emph{block
graph}, if every block in $G$ is a complete graph. Notice that block graphs are a proper superfamily of tree graphs and a proper subfamily of chordal graphs. Suppose $G$ has blocks $B_1, B_2,\ldots,B_x$ and cut vertices $c_1,
c_2,\ldots,c_y$. We define the \emph{block-cut-vertex graph} $G^* =
(V, E)$ of $G$, where \vspace{-23pt}

\begin{align*}
V(G^*) &= \setof{B_1,B_2,\ldots,B_x,c_1, c_2,\ldots,c_y}; \text{~and}\\
E(G^*) &= \setof{(B_i,c_j) \mid c_j \in B_j, 1\le i \le x, 1\le j
\le y }.
\end{align*}

Consequently, \rgb{the} graph $G^*$ is a tree and the leaves in $G^*$ are
precisely the blocks with exactly one cut-vertex in $G$. A block
containing exactly one cut-vertex in $G$ is called \emph{\rgb{a pendant}
block}. It should be noted that, by using the depth first search algorithm, one can recognize the block graphs and construct the block-cut-vertex graphs $G^*$, both in $O(n+m)$ time~\cite{Aho74}. Figure~\ref{fig:block_graph} shows an illustrative example, in which Figure~\ref{fig:block_graph}$(b)$ depicts the corresponding block-cut-vertex graph $G^*$ for the block graph $G$ in Figure~\ref{fig:block_graph}$(a)$. Clearly, $G$ has $8$ blocks $B_1, B_2,\ldots,B_8$ and $6$ cut vertices $c_1,
c_2,\ldots,c_6$. Moreover, the pendant blocks of $G$ are $B_1, B_2, B_4$, $B_5$, and $B_7$.

\subsection{The algorithm}
\label{subsection:algo-dynamic}
In this subsection, given the block-cut-vertex structure of a weighted block graph $G$, we propose an $O(n)$-time algorithm for finding a minimum-weight paired-dominating set of $G$. Before describing the approach in detail, four notations $D(H, u)$, $P(H, u)$, $P'(H, u)$, and $\bar{P}(H, u)$ are defined below, where $H$ is a subgraph of $G$ and $u \in V(H)$. The notations are introduced for the purpose of describing the recursive formulations used in developing dynamic programming algorithms.
\vspace{18pt}
\\
\mbox{} \hspace{10pt} $D(H, u)$ : A minimum-weight dominating set of $H$ \rgb{containing $u$}, and $H[D(H, u)-{u}]$ \\
\mbox{} \hspace{63.5pt} has a perfect matching. \vspace{4pt}\\
\mbox{} \hspace{10.5pt} $P(H, u)$ : A minimum-weight paired-dominating set of $H$ containing $u$. \vspace{4pt}\\
\mbox{} \hspace{7.8pt} $P'(H, u)$ : A minimum-weight paired-dominating set of $H$ \rgb{not containing} $u$.\vspace{4pt}\\
\mbox{} \hspace{10.5pt} $\bar{P}(H, u)$ : A minimum-weight paired-dominating set of $H-u$, and $u$ is not dominated \\ \mbox{} \hspace{63.5pt} by $\bar{P}(H, u)$. \vspace{18pt}

\noindent Clearly, either $P(G, u)$ or $P'(G, u)$ is a minimum-weight paired-dominating set of $G$. For ease of subsequent discussion, $D(H, u)$, $P(H, u)$, $P'(H, u)$, and $\bar{P}(H, u)$ are called a $\kappa_1$-\emph{paired-dominating set}, $\kappa_2$-\emph{paired-dominating set}, $\kappa_3$-\emph{paired-dominating set}, and $\kappa_4$-\emph{paired-dominating set} of $H$ with respect to $u$, respectively. Suppose that $H$ is a weighted block graph and $B$ is a block of $H$ with $V(B) = \setof{u_1, u_2,\ldots, u_k}$. The following lemma shows an useful property which help us design efficient algorithms. 

\begin{figure}[thb]
\centerline{\input{union}} \vspace{15pt}\caption{A weighted block graph $H = B \cup G_1 \cup G_2 \cup \ldots \cup G_k$.} \vspace{5pt} \label{fig:union}
\end{figure}
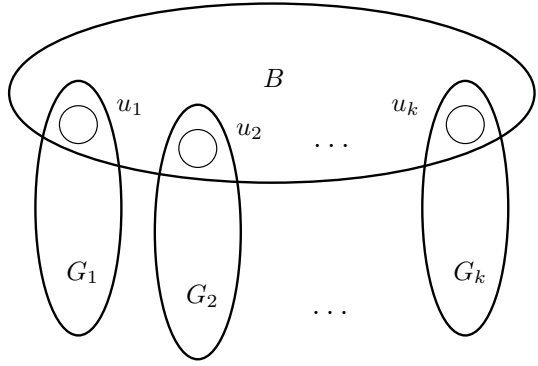

\begin{lemma}
\label{lemma:disjoint}
Suppose that $H$ is a weighted block graph and $B$ is a block of $H$ with 
$V(B) = \setof{u_1, u_2,\ldots, u_k}$. If $G_i$ is a maximal connected subgraph in $(H -B) \cup \setof{u_i}$ for $1 \le i \le k$, then $G_i$ and $G_j$ have disjoint vertex sets for $i \not = j$.
\end{lemma}
\begin{proof}
Suppose to the contrary that $V(G_i) \cap V(G_j) \not = \emptyset$. Then $G[B \cup G_i \cup G_j]$ is a connected subgraph of $H$ without a cut-vertex, this contradicts our assumption that $B$ is a maximal connected subgraph without a cut-vertex.  
\end{proof}

\medskip
Refer to Figure~\ref{fig:union} for an illustrative example. In order to obtain a minimum-weight paired-dominating set of $G$, we use dynamic programming to iteratively determine $D(H, u_1)$, $P(H, u_1)$, $P'(H, u_1)$, and $\bar{P}(H, u_1)$ in a bottom-up manner. \rgb{One block is considered in each iteration of the loop.} Suppose the dominating sets $D(G_i, u_i)$, $P(G_i, u_i)$, $P'(G_i, u_i)$, and $\bar{P}(G_i, u_i)$ have been determined in the previous iterations and are recorded in $u_i$ for $1 \le i \le k$. \rgb{We shall show that $D(H, u_1)$, $P(H, u_1)$, $P'(H, u_1)$, and $\bar{P}(H, u_1)$ can all be determined in $O(k)$ time in Section~\ref{section:efficient-implementation}. With the aid of this result, we now propose the main algorithm of this paper. Notice that during the computation, the block-cut-vertex structure $G^*$ of the block graph
$G$ can be exploited to get the relationship among blocks, which help us to apply dynamic programming.}


The algorithm first sets the \emph{current graph} $G' = G$ and the set of \emph{processed blocks} $W = \emptyset$. Further, it initially assigns  $D(G[\{v\}],v ) = \setof{v}$, $P(G[\{v\}],v ) = \bigtriangleup$, $P'(G[\{v\}],v ) = \bigtriangleup$ and $\bar P(G[\{v\}],v ) = \emptyset$ to each vertex $v \in V(G)$. Specially, we use $\bigtriangleup$ to denote the empty set with a weight of infinity, i.e., $\bigtriangleup = \emptyset$ and $w(\bigtriangleup) = \infty$. The algorithm then iteratively processes blocks in the repeat loop. During each iteration of the loop, we remove a pendant block $B$ in the current graph $G'$ and determines the dominating sets $D(H, u)$, $P(H, u)$, $P'(H, u)$, and $\bar{P}(H, u)$, where $u$ is the cut vertex and $H$ is the connected component containing $u$ in $G[B \cup W]$. After the execution of the repeat loop, we have only one block left, i.e., the current graph $G'$ is a block and $G^*$ is a vertex. With the information determined in the repeat loop, we now can find the two paired-dominating sets $P(G, u)$ and $P'(G, u)$, where $u$ is an arbitrary vertex in $G'$. Finally, the output $S$ is selected from $P(G, u)$ and $P'(G, u)$ based on the weights of the sets. The steps of the algorithm are detailed below. 

%
%

\vspace{8pt}

\begin{algorithm}
\caption{Finding a paired-dominating set on weighted block graphs}
\begin{algorithmic} [1]
\medskip
\baselineskip 15.8pt
\REQUIRE A weighted block graph $G$ with the block-cut-vertex structure $G^*$ of $G$.
\ENSURE A minimum-weight paired-dominating set $S$ of $G$.

\STATE let $G' \leftarrow G$ and $W \leftarrow \emptyset$;

\FOR{each $v \in V(G)$} 
\STATE let $D(G[\{v\}],v ) \leftarrow \setof{v}$ and $P(G[\{v\}],v ) \leftarrow \bigtriangleup$;
\STATE let $P'(G[\{v\}],v ) \leftarrow \bigtriangleup$ and $\bar{P}(G[\{v\}],v ) \leftarrow \emptyset$;
\ENDFOR

\REPEAT

\STATE arbitrarily choose a leaf $v_B$ in $G^*$;

\STATE let $B$ be the corresponding pendant block of $v_B$ in $G'$; \\ suppose that $V(B) = \setof{u_1, u_2,\ldots, u_k}$, where $u_1$ is the cut vertex and $G_i$ is the connected component in $G[W]$ such that $V(B) \cap V(G_i) = {u_i}$ for $1 \le i \le k$;

\STATE let $H \leftarrow B \cup G_1 \cup G_2 \cup \ldots \cup G_k$;

\STATE find $D(H, u_1), P(H, u_1), P'(H, u_1), \bar P(H, u_1)$ by using the dominating sets $D(G_i, u_i)$, $P(G_i, u_i)$, $P'(G_i, u_i)$, and $\bar{P}(G_i, u_i)$ determined in the previous iterations;

\STATE record the results $D(H, u_1), P(H, u_1), P'(H, u_1), \bar P(H, u_1)$ in vertex $u_1$;
    
\STATE let $G' \leftarrow G' - \setof{u_2,\ldots, u_k}$ and $W \leftarrow W \cup B$;

\STATE suppose $v_c$ is the neighbor of $v_B$ in $G^*$; \\ let $G^* \leftarrow G^* - \setof{v_B,v_c}$ if $v_c$ is a leaf in $G^* - v_B$, and let $G^* \leftarrow G^* - v_B$ otherwise;

\UNTIL{$G^*$ itself is a vertex}

\STATE find $P(G, u)$ and $P'(G, u)$, where $u$ is an arbitrary vertex in $G'$;
    
\STATE let $S \leftarrow P(G, u)$ if $w(P(G, u)) < w(P'(G, u))$, and let $S \leftarrow P'(G, u)$ otherwise;

\RETURN $S$.
\end{algorithmic}
\end{algorithm}
\baselineskip \skippt

\vspace{-5pt}

\begin{figure}[p]
\centerline{\input{example}} \vspace{2pt}\caption{The intermediate execution steps of Algorithm 1. The blocks in $G$ are removed with
respect to the ordering $B_1, B_2,\ldots, B_8$.} \vspace{5pt} \label{fig:example}
\end{figure}
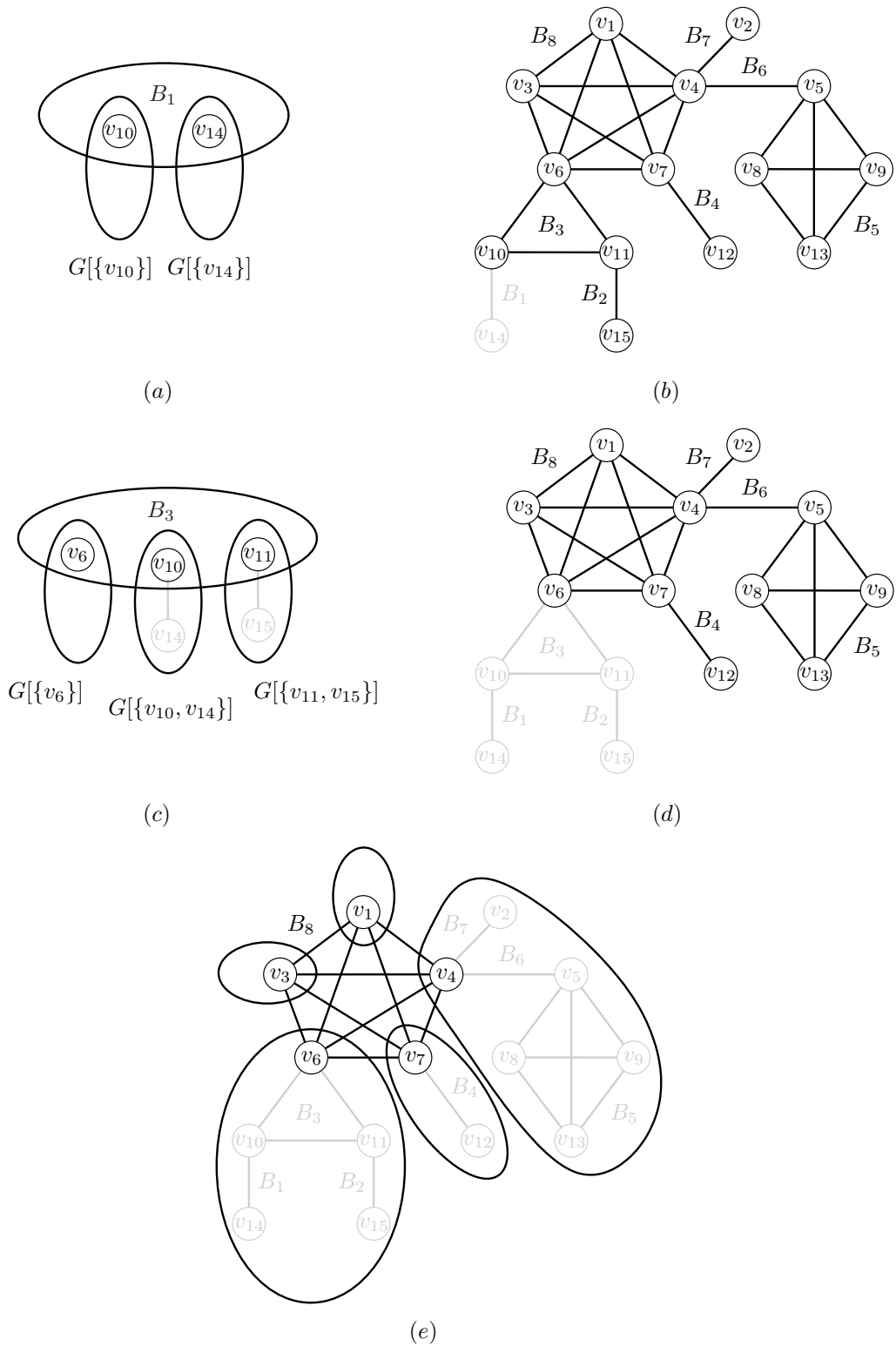

For an illustrative example, consider the block graph $G$ in Figure~\ref{fig:block_graph}. In the
beginning, the algorithm sets the default values to each vertex $v$ in $G$.  Then, by the rules of removing blocks and recording results, one block is removed from $G'$ for each iteration of the repeat loop. The blocks in $G$ are removed with respect to the ordering $B_1, B_2,\ldots, B_8$.
Figure~\ref{fig:example}$(a)$ depicts the case that block $B_1 = G[\setof{v_{10},v_{14}}]$ is selected in the first iteration. One can see that $H = G[\setof{v_{10},v_{14}}] \cup  G[\setof{v_{10}}] \cup  G[\setof{v_{14}}]$. Four dominating sets $D(H, v_{10}), P(H, v_{10}), P'(H, v_{10}), \bar P(H, v_{10})$ are recorded in vertex $v_{10}$ by the rules of determining and recording results in Steps 10 and 11. Please refer to Figure~\ref{fig:example}$(b)$ for the result of removing $v_{14}$ from $G'$.

Similarly, Figure~\ref{fig:example}$(c)$ depicts the case that block $B_3 = G[\setof{v_6,v_{10},v_{11}}]$ is selected in the third iteration with $H = G[\setof{v_6,v_{10},v_{11}}]  \cup  G[\setof{v_6}] \cup  G[\setof{v_{10},v_{14}}] \cup  G[\setof{v_{11},v_{15}}]$. Again, by the rules of determining and recording results, vertices $v_{10}$ and $v_{11}$ were removed from $G'$ and the corresponding results are recorded in vertex $v_6$. Refer to Figure~\ref{fig:example}$(d)$ for an illustrative example. After removing blocks $B_1, B_2,\ldots, B_7$, we have exactly one block $B_8$ left in $G'$, i.e., $G^*$ now is a vertex, then the algorithm exits the repeat loop. In Step 15, the two dominating sets $P(G, v_1)$ and $P'(G, v_1)$ are determined by using a similar method of the arguments in Steps 10 and 11. Refer to Figure~\ref{fig:example}$(e)$ for an illustrative example. Clearly, either $P(G, u)$ or $P'(G, u)$ is a minimum-weight paired-dominating set $S$ of $G$ depending on which has the smaller total weight. 

In next section, given the dominating sets $D(G_i, u_i)$, $P(G_i, u_i)$, $P'(G_i, u_i)$, and $\bar{P}(G_i, u_i)$ for $1 \le i \le k$, four dynamic programming procedures are proposed in Subsections~\ref{subsection:finding-D(H, u_1)}--\ref{subsection:finding-bar_P(H, u_1)}, which can determine $D(H, u_1)$, $P(H, u_1)$, $P'(H, u_1)$, and $\bar{P}(H, u_1)$ in $O(k)$ time, respectively. Clearly, the proposed procedures ensure the correctness of the algorithm. For the complexity analysis, suppose that $G$ has blocks $B_1, B_2,\ldots,B_x$. Since the dynamic programming procedures can be completed in $O(k)$ time, Steps 10 and 15 can be implemented in $O(|V(B_1)|+ |V(B_2)|+\ldots+|V(B_x)|)$ time. Recall that all the vertices in $B_i$ are deleted from $G'$ except the cut vertex in each iteration of the repeat loop. This implies that $|V(B_1)|+ |V(B_2)|+\ldots+|V(B_x)| = n + (x -1)$. 

Notice that by using the depth first search algorithm, one can computer a vertex ordering $u_1, u_2,\ldots, u_h$ of a $h$-vertices tree graph $T$ in $O(h)$ time such that $u_i$ is a leaf in $T[\setof{u_1,u_2,\ldots,u_i}]$ for $1 \le i \le h$. Therefore, with an $O(x)$-time preprocessing of $G^*$, it takes $O(1)$ time to implement Step 7 for each iteration of the repeat loop, i.e., determining $v_B$. Meanwhile, since $G$ has at most $n-1$ blocks, i.e., $x \le n - 1$, the repeat loop and Step 15 can be done, both in $O(n)$ time. Further, the other steps can be completed in $O(n)$ time as well. Consequently, we obtain the main result of this paper. 
\begin{theorem}
\label{theorem-1}
Given a weighted block graph $G$ with the block-cut-vertex structure $G^*$ of $G$, a paired-dominating set of $G$ can be determined by Algorithm 1 in $O(n)$ time. 
\end{theorem}
\section{Finding $D(H, u_1)$, $P(H, u_1)$, $P'(H, u_1)$, and $\bar{P}(H, u_1)$}\label{section:efficient-implementation}
Suppose that $H$ is a weighted block graph and $B$ is a block of $H$ with $V(B) = \setof{u_1, u_2,\ldots, u_k}$. For $1 \le i \le k$, let $G_i$ be a maximal connected subgraph in $(H - B) \cup \setof{u_i}$. By Lemma~\ref{lemma:disjoint}, we have $V(G_i)\cap V(G_j) = \emptyset$ for $i \not = j$, refer to Figure~\ref{fig:union} for an illustrative example. \rgb{In this section, given the dominating sets $D(G_i, u_i)$, $P(G_i, u_i)$, $P'(G_i, u_i)$, and $\bar{P}(G_i, u_i)$ for $1 \le i \le k$, we shall show that the four dominating sets $D(H, u_1)$, $P(H, u_1)$, $P'(H, u_1)$, and $\bar{P}(H, u_1)$ can be determined in $O(k)$ time.}

First, some notations are introduced below, for the purpose of describing the procedures. For a set $S$ of sets of vertices, $F(S)$ denotes the set with minimum weight in $S$. Let $S_i^*$ be the set of vertices such that $S_i^* = F(\setof{D(G_i, u_i), P(G_i, u_i), P'(G_i, u_i), \bar P(G_i, u_i)})$ for $2 \le i \le k$. We use $\alpha$ to denote the index in $\setof{2, 3, \ldots, k}$ such that $S_\alpha^* \not = D(G_\alpha, u_\alpha)$ and $w(D(G_\alpha, u_\alpha)) - w(S_\alpha^*)$ is minimized, and $\beta$ to denote the index in $\setof{2, 3, \ldots, k}$ such that $S_\beta^* = D(G_\beta, u_\beta)$ and $w(F(\setof{P(G_\beta, u_\beta), P'(G_\beta, u_\beta), 
\bar P(G_\beta, u_\beta)})) - w(S_\beta^*)$ is minimized. Further, let $r$ denote the number of $S_i^*$ such that $S_i^* = D(G_i, u_i)$, i.e., $r = |\setof{S_i^* \mid S_i^* =D(G_i, u_i) \text{~and~} 2 \le i \le k}|$.

\subsection{Determination of $D(H, u_1)$}\label{subsection:finding-D(H, u_1)}
We first recall that $D(H,u_1)$ is a minimum-weight dominating set of $H$ \rgb{containing $u_1$}, and $H[D(H, u_1)-{u_1}]$ has a perfect matching. By the definition of $D(H, u_1)$, the only potential candidate for being a dominating set of $G_1$ is $D(G_1,u_1)$. Hence, in order to obtain $D(H,u_1)$, we first construct a dominating set $X = D(G_1,u_1) \cup S_2^* \cup S_3^* \cup \ldots \cup S_k^*$. We will show that if $r$ is even, then $S = X$ is a $\kappa_1$-dominating set of $H$ with respect to $u_1$. Otherwise, for the purpose of satisfying the requirement that $H[X - u_1]$ has a perfect matching with minimum cost, we can either replace $S_\alpha^*$ with $D(G_\alpha,u_\alpha)$, or replace $S_\beta^*$ with $F(\setof{P(G_\beta, u_\beta), P'(G_\beta, u_\beta), \bar P(G_\beta, u_\beta)})$. For the former case, a dominating set $X^+ = (X - S_\alpha^*) \cup D(G_\alpha, u_\alpha)$
is created. On the other hand, a dominating set $X^- = (X - S_\beta^*) \cup F(\setof{P(G_\beta, u_\beta), P'(G_\beta, u_\beta), \bar P(G_\beta, u_\beta)})$ is built for the latter case. The output $S = F(\setof{X^+, X^-})$ is selected from $X^+$ and $X^-$ based on the weights of the sets. Similarly, we will show that $S$ is a $\kappa_1$-dominating set of $H$ with respect to $u_1$. The following is a formal description of the procedure.

\vspace{12pt}
\floatname{algorithm}{Procedure}
\begin{algorithm}
\caption{Finding $D(H, u_1)$}
\begin{algorithmic} [1]
\medskip
\baselineskip 19pt
\REQUIRE A weighted block graph $H$ and a block $B$ of $H$ with 
$V(B) = \setof{u_1, u_2,\ldots, u_k}$. 
\\ \hspace{18pt} Dominating sets $D(G_i, u_i)$, $P(G_i, u_i)$, $P'(G_i, u_i)$, and $\bar{P}(G_i, u_i)$ for $1 \le i \le k$.

\ENSURE A $\kappa_1$-paired-dominating set $S$ of $H$ with respect to $u_1$.


\STATE determine $S_i^*$ for $2 \le i \le k$;

\STATE determine $\alpha, \beta$, and $r$;  

\STATE let $X \leftarrow D(G_1,u_1) \cup S_2^* \cup S_3^* \cup \ldots \cup S_k^*$;

\STATE let $X^+ \leftarrow (X - S_\alpha^*) \cup D(G_\alpha, u_\alpha)$; 

\STATE let $X^- \leftarrow (X - S_\beta^*) \cup F(\setof{P(G_\beta, u_\beta), P'(G_\beta, u_\beta), \bar P(G_\beta, u_\beta)})$;

\STATE {\bf if} $r$ is even, {\bf then} let $S \leftarrow X$; {\bf otherwise}, let $S \leftarrow F(\setof{X^+, X^-})$;

\RETURN $S$.
\end{algorithmic}
\end{algorithm}
\baselineskip \skippt

\vspace{-10pt}
\begin{lemma}
\label{lemma:finding-D(H, u_1)}
Given the dominating sets $D(G_i, u_i)$, $P(G_i, u_i)$, $P'(G_i, u_i)$, and $\bar{P}(G_i, u_i)$ for $1 \le i \le k$, Procedure~2 outputs a $\kappa_1$-paired-dominating set $S$ of $H$ with respect to $u_1$ in $O(k)$ time.
\end{lemma}
\begin{proof}
\rgb{We first introduce data structures which enable us to compute $D(H, u_1)$ in $O(k)$ time. For each element of array $select$, with $2 \le i \le k$, $select[i]$ is used to denote the selection of vertices set for $S_i^*$, i.e., $select[i] = 1$ if $S_i^* = D(G_i, u_i),\ldots, select[i] = 4$ if $S_i^* = \bar{P}(G_i, u_i)$. Further, variable $s$ is used to denote the selection of vertices set for $S$, i.e., $s = 1$ if $S = X$, $s = 2$ if $S = X^+$, and $s = 3$ if $S = X^-$. Meanwhile, variable $w$ is used to denote the weight of $S$. With the aid of above data structures and variables $\alpha, \beta$, and $r$, the procedure certainly can be implemented in $O(k)$ time.} To prove that $S$ is a $\kappa_1$-paired-dominating set of $H$ with respect to $u_1$, it suffices to show that the output $S$ is a minimum-weight dominating set of $H$ such that $u_1 \in S$ and $H[S - u_1]$ has a perfect matching. By the definition of $D(H, u_1)$, the only potential candidate for being a dominating set of $G_1$ is $D(G_1,u_1)$. Hence, we have $D(G_1, u_1) \subseteq S$. Since $u_1 \in D(G_1, u_1)$ and $B$ is a clique, all the three sets $X, X^+$ and $X^-$ are dominating \rgb{sets} of $H$. Therefore, it remains to show that the weight $w(S)$ of $S$ is minimized subject to the condition that $H[S - u_1]$ contains a perfect matching.

Notice that, for $2 \le i \le k$, both $G_i[D(G_i, u_i)- u_i]$ and $G_i[P(G_i, u_i)]$ contain perfect matchings and $u_i \not \in P'(G_i, u_i) \cup \bar{P}(G_i, u_i)$. \rgb{Therefore, to satisfy the condition that $H[S - u_1]$ contains a perfect matching, the cardinality of $\setof{u_i \mid D(G_i, u_i) \in S \text{~and~} 2 \le i \le k}$ must be even. Hence, if $r$ is even, then $H[X-u_1]$ contains a perfect matching and the weight $w(X)$ of $X$ is minimized, as an immediate consequence of the selections of $S_i^*$. 
Next, suppose that $r$ is odd. To minimize the cost, we can replace one $S_i^*$ with $D(G_i, u_i)$  or replace one $S_j^*$ with $P(G_j, u_j), P'(G_j, u_j)$, or  $\bar{P}(G_j, u_j)$, where $S_i^* \not = D(G_i, u_i)$, $S_j^* = D(G_j, u_j)$, and $2 \le i,j \le k$. For the former case, a dominating set $X^+ = (X - S_\alpha^*) \cup D(G_\alpha, u_\alpha)$ is created. On the other hand, a dominating set $X^- = (X - S_\beta^*) \cup F(\setof{P(G_\beta, u_\beta), P'(G_\beta, u_\beta), \bar P(G_\beta, u_\beta)})$ is built for the latter case.} And, we select $S$ from $X^+$ and $X^-$ based on the weights of the sets. As a consequence of selections of $S_\alpha^*$, $S_\beta^*$, and $S_i^*$ for $2 \le i \le k$, one can verify that $S = F(\setof{X^+, X^-})$ is a minimum-weight dominating set of $H$ such that $H[S - u_1]$ contains a perfect matching. 
\end{proof}
\subsection{Determination of $P(H, u_1)$}\label{subsection:finding-P(H, u_1)}
Notice that $P(H,u_1)$ is a minimum-weight dominating set of $H$ containing $u_1$ and $H[P(H,u_1)]$ has a perfect matching. Therefore, either $D(G_1,u_1)\subseteq P(H,u_1)$ or $P(G_1,u_1)\subseteq P(H,u_1)$ is a dominating set of $G_1$. In order to obtain $P(H,u_1)$, we construct the six dominating sets $X, X^+$, $X^-$, $Y, Y^+$, and $Y^-$ of $H$. The dominating sets $X, X^+$ and $X^-$ are created for the situation when $D(H,u_1)$ is a dominating set of $G_1$. Meanwhile, the dominating sets $Y, Y^+$ and $Y^-$ are built for  the situation when $P(G_1,u_1)$ is a dominating set of $G_1$, where $Y = P(G_1,u_1)\cup S_2^* \cup S_3^* \cup \dots \cup S_k^*$, $Y^+ = (Y-S_\alpha^*) \cup D(G_\alpha, u_\alpha)$, and $Y^- = (Y-S_\beta^*)\cup F(\{P(G_\beta, u_\beta),P'(G_\beta, u_\beta),\bar P(G_\beta, u_\beta)\})$. 

If $r$ is even, then \rgb{all} the induced subgraphs $H[X^+], H[X^-]$ and $H[Y]$ contain perfect matchings. The output $S = F(\setof{X^+, X^-,Y})$ is selected from $X^+, X^-$ and $Y$ based on the weights of the sets. We will show that $S$ is a $\kappa_2$-dominating set of $H$ with respect to $u_1$. Similarly, if $r$ is odd, then all the induced subgraphs $H[X], H[Y^+]$ and $H[Y^-]$ contain perfect matchings. And, we will show that the output $S = F(\setof{X, Y^+,Y^-})$ is a $\kappa_2$-dominating set of $H$ with respect to $u_1$ in this situation. The procedure is detailed below.

\vspace{10pt}
\floatname{algorithm}{Procedure}
\begin{algorithm}
\caption{Finding $P(H, u_1)$}
\label{algorithm:finding-p}
\begin{algorithmic} [1]
\medskip
\baselineskip 19pt
\REQUIRE A weighted block graph $H$ and a block $B$ of $H$ with 
$V(B) = \setof{u_1, u_2,\ldots, u_k}$. 
\\ \hspace{18pt} Dominating sets $D(G_i, u_i)$, $P(G_i, u_i)$, $P'(G_i, u_i)$, and $\bar{P}(G_i, u_i)$ for $1 \le i \le k$. 

\ENSURE A $\kappa_2$-paired-dominating set $S$ of $H$ with respect to $u_1$.


\STATE determine $S_i^*$ for $2 \le i \le k$;

\STATE determine $\alpha, \beta$, and $r$;

\STATE find the dominating sets $X, X^+$ and $X^-$ as described in Procedure $2$;  

\STATE let $Y \leftarrow P(G_1,u_1)\cup S_2^* \cup S_3^* \cup \dots \cup S_k^*$; 
  
\STATE let $Y^+ \leftarrow (Y-S_\alpha^*) \cup D(G_\alpha, u_\alpha);$

\STATE let $Y^- \leftarrow (Y-S_\beta^*)\cup F(\{P(G_\beta, u_\beta),P'(G_\beta, u_\beta),\bar P(G_\beta, u_\beta)\})$;

\STATE {\bf if} $r$ is even, {\bf then} let $S \leftarrow F(\setof{X^+, X^-,Y})$; {\bf otherwise}, let $S \leftarrow F(\setof{X, Y^+, Y^-})$;

\RETURN $S$.

\end{algorithmic}
\end{algorithm}
\baselineskip \skippt

\vspace{-8pt}
\begin{lemma}
\label{lemma:finding-P(H, u_1)}
Given the dominating sets $D(G_i, u_i)$, $P(G_i, u_i)$, $P'(G_i, u_i)$, and $\bar{P}(G_i, u_i)$ for $1 \le i \le k$, Procedure~3 outputs a $\kappa_2$-paired-dominating set $S$ of $H$ with respect to $u_1$ in $O(k)$ time.
\end{lemma}
\begin{proof}
By using a similar method of the arguments in Lemma~\ref{lemma:finding-D(H, u_1)}, one can show that the procedure can be completed in $O(k)$ time. To prove the correctness of the procedure, it suffices to show that the output $S$ is a minimum-weight dominating set of $H$ such that $u_1 \in S$ and $H[S]$ contains a perfect matching. Further, since $v_1 \in D(G_1, u_1) \cap P(G_1, u_1)$ and $B$ is a clique, all $X, X^+, X^-, Y, Y^+$, and $Y^-$ are dominating sets of $H$. Thus, it remains to show that  the weight $w(S)$ of $S$ is minimized subject to the condition that $H[S]$ contains a perfect matching. 

Notice that, for $2 \le i \le k$, both $G_i[D(G_i, u_i)- u_i]$ and $G_i[P(G_i, u_i)]$ contain perfect matchings and $u_i \not \in P'(G_i, u_i) \cup \bar{P}(G_i, u_i)$. We first consider the situation when $r$ is even. For the case when $D(G_1,u_1)$ is a dominating set of $G_1$, in order to satisfy the condition that $H[X]$ contains a perfect matching with minimum cost, we can either replace $S_\alpha^*$ with $D(G_\alpha,u_\alpha)$ or replace $S_\beta^*$ with $F(\setof{P(G_\beta, u_\beta), P'(G_\beta, u_\beta), \bar P(G_\beta, u_\beta)})$. Thus, $X^+$ and $X^-$ are the two potential candidates for $S$ when $D(G_1,u_1)\subseteq P(H,u_1)$. For the case when $P(G_1,u_1)$ is a dominating set of $G_1$, $H[Y]$ contains a perfect matching. We select $S = F(\setof{X^+, X^-,Y})$ from $X^+$, $X^-$ and $Y$ based on the weights of the sets. As a consequence of selections of $S_\alpha^*$, $S_\beta^*$, and $S_i^*$ for $2 \le i \le k$, one can verify that the output $S$ is a minimum-weight dominating set of $H$ such that $H[S]$ contains a perfect matching. Using a similar method of the above arguments, one can show that the correctness also holds for the situation when $r$ is odd. 
\end{proof}

\subsection{Determination of $P'(H, u_1)$}\label{subsection:finding-P'(H, u_1)}

Recall that $P'(H,u_1)$ is a minimum-weight dominating set of $H$ not containing $u_1$ and $H[P'(H,u_1)]$ has a perfect matching. 
Therefore, either $P'(G_1,u_1)\subseteq P'(H,u_1)$ or $\bar{P}(G_1,u_1)\subseteq P'(H,u_1)$ is a dominating set of $G_1$. For ease of subsequent discussion, we consider the two cases $P'(G_1,u_1)\subseteq P'(H,u_1)$ and $\bar{P}(G_1,u_1)\subseteq P'(H,u_1)$, respectively, in the rest of this subsection. More concretely, a paired-dominating set $Q_1$ is created for the former situation. Meanwhile, a paired-dominating set $Q_2$ is built for the latter situation. Clearly, $P'(H, u_1)$ can be selected from $Q_1$ and $Q_2$ based on the weights of the sets.

\subsubsection{Finding $Q_1$}\label{subsubsection:finding-Q_1}
\vspace{-3pt}

Below we present an $O(k)$-time procedure for finding $Q_1$. The procedure solves the problem by considering eight cases $C_1, C_2, \ldots, C_8$ depending on $S_i^*$ and $r$. For $1 \le i \le 8$, the case $C_i = (c_1, c_2,c_3,c_4, c_5)$ is an ordered 5-tuple. For $1 \le j \le 5$, $c_j = 1$ if condition $D_j$ holds, and $c_j = 0$ otherwise. \rgb{ Further, $c_j = ``*$'' means ``do not care'', i.e., condition $D_j$ is not a factor in this case. }The five conditions $D_1,D_2,\dots,D_5$ are defined as follows:
\vspace{12pt}
\\
\mbox{} \hspace{55pt} $D_1$ : $S_i^* = P(G_i,u_i)$ for some $2 \le i \le k$. \\
\mbox{} \hspace{55pt} $D_2$ : $r$ is odd. \\
\mbox{} \hspace{55pt} $D_3$ : $r$ is equal to $1$.\\
\mbox{} \hspace{55pt} $D_4$ : $r$ is equal to $0$.\\
\mbox{} \hspace{55pt} $D_5$ : $S_i^* = \bar{P}(G_i,u_i)$ for some $2 \le i \le k$. \vspace{10pt}

\noindent Then, we define the cases $C_1 = (1,1,*,*,*)$,  $C_2 = (1,0,*,*,*)$, $C_3 = (0,1,1,*,1)$, $C_4 = (0,1,1,*,0)$, $C_5 = (0,1,0,*,*)$, $C_6 = (0,0,*,1,1)$, $C_7 = (0,0,*,1,0)$, and $C_8 = (0,0,*,0,*)$. For example, case $C_1$ represents the situation when there exists an index $\ell$ such that $S_\ell^* = P(G_\ell,u_\ell)$ with $2 \le \ell \le k$ and $r$ is an odd number. Further, case $C_7$ represents the situation when there exists no index $\ell$ such that $S_\ell^* = P(G_\ell,u_\ell)$, or $S_\ell^* = \bar{P}(G_\ell,u_\ell)$ and $r = 0$, i.e., $S_i^* = P'(G_i,u_i)$ for $2 \le i \le k$. Moreover, one can verify that all the possible combinations of the five conditions have been considered.

Next, some notations and paired-dominating sets are introduced. Let ${\alpha'}$ be the index in $\setof{2, 3, \ldots, k} - \setof{\alpha}$ such that $S_{\alpha'}^* \not = D(G_{\alpha'}, u_{\alpha'})$ and $w(D(G_{\alpha'}, u_{\alpha'})) - w(S_{\alpha'}^*)$ is minimized. Let $\gamma$ be the index in $\setof{2, 3, \ldots, k}$ \rgb{such that $S_\gamma^* \not = P(G_\gamma, u_\gamma)$ and $w(P(G_\gamma, u_\gamma)) - w(S_\gamma^*)$ is minimized.} Let $I = \setof{i \mid S_i^* = \bar{P}(G_i, u_i) \text{~and~} 2 \le i \le k}$. We define the following paired-dominating sets of $H$, which are the potential candidates for $Q_1$.
\vspace{10pt}
\\
\mbox{} \hspace{55pt} $Z_1 \hspace{3.4pt}= P'(G_1,u_1)\cup S_2^* \cup S_3^* \cup \ldots \cup S_k^*$. \\
\mbox{} \hspace{55pt} $Z_1^+ \hspace{0.2
pt}= (Z_1-S_\alpha^*) \cup D(G_\alpha,u_\alpha).$\\
\mbox{} \hspace{55pt} \hspace{0.1pt}$Z_1^- = (Z_1-S_\beta^*)\cup F(\{P(G_\beta, u_\beta),P'(G_\beta, u_\beta),\bar P(G_\beta, u_\beta)\})$.\\
\mbox{} \hspace{55pt} $T_1 \hspace{4.5pt} = (Z_1-S_\gamma^*) \cup P(G_\gamma, u_\gamma)$.\\
\mbox{} \hspace{55pt} $T_2 \hspace{4.5pt} = (Z_1-S_\alpha^*-S_{\alpha'}^*) \cup D(G_\alpha,u_\alpha)\cup D(G_{\alpha'}, u_{\alpha'})$.\\
\mbox{} \hspace{55pt} $T_3 \hspace{4.5pt} = (Z_1-\cup_{i\in I}S_i^*)\cup (\cup_{i\in I}P'(G_i,u_i))$.\\
\mbox{} \hspace{55pt} $T_4 \hspace{4.5pt} = (Z_1-S_\beta^*)\cup P(G_\beta,u_\beta)$.\\
\mbox{} \hspace{55pt} $T_5 \hspace{4.5pt} = (Z_1-S_\beta^*)\cup F(\{P(G_\beta, u_\beta), P'(G_\beta, u_\beta)\})$.\\
\mbox{} \hspace{55pt} $T_6 \hspace{4.5pt} =  (Z_1-S_\gamma^*-S_\beta^*)\cup P(G_\gamma,u_\gamma)\cup F(\{P'(G_\beta,u_\beta), \bar P(G_\beta,u_\beta)\})$.\\
\mbox{} \hspace{55pt} $T_7 \hspace{4.5pt} = (Z_1-S_\gamma^*-S_\beta^*)\cup P(G_\gamma, u_\gamma) \cup \bar P(G_\beta,u_\beta)$.\\
\mbox{} \hspace{55pt} $T_8 \hspace{4.5pt} = (Z_1-\cup_{i\in I}S_i^*-S_\beta^*)\cup (\cup_{i\in I}P'(G_i,u_i))\cup P'(G_\beta,u_\beta)$. \vspace{11pt}

\noindent As mentioned earlier, we solve the problem by considering the eight cases $C_1, C_2, \ldots, C_8$. The relations between the cases $C_1, C_2, \ldots, C_8$ and the dominating sets $Z_1$, $Z_1^+$, $Z_1^-$, $T_1, \dots, T_8$ are detailed in Procedure~\ref{algorithm:finding-Q_1}. We will prove its correctness and analyze its running time in Lemma~\ref{lemma:finding-Q_1}.

\vspace{12pt}
\floatname{algorithm}{Procedure}
\begin{algorithm}
\caption{Finding $Q_1$}
\begin{algorithmic} [1]
\label{algorithm:finding-Q_1}
\medskip
\baselineskip 19pt
\REQUIRE A weighted block graph $H$ and a block $B$ of $H$ with 
$V(B) = \setof{u_1, u_2,\ldots, u_k}$. 
\\ \hspace{18pt} Dominating sets $D(G_i, u_i)$, $P(G_i, u_i)$, $P'(G_i, u_i)$, and $\bar{P}(G_i, u_i)$ for $1 \le i \le k$. 

\ENSURE A minimum-weight dominating set $S$ of $H$ such that $u_1 \not \in S$, $H[S]$ has a perfect \\ \hspace{32.7pt}matching, and $P'(G_1,u_1)\subseteq S$.


\STATE determine the paired-dominating sets $Z_1$, $Z_1^+$, $Z_1^-$, $T_1, \dots, T_8$;

\STATE {\bf if} $C_1$ or $C_5$ holds, {\bf then} let $S \leftarrow F(\setof{Z_1^+, Z_1^-})$;

\STATE {\bf if} $C_2$ or $C_7$ or $C_8$ holds, {\bf then} let $S \leftarrow Z_1$;

\STATE {\bf if} $C_3$ holds, {\bf then} let $S \leftarrow F(\setof{Z_1^+, T_4, T_6, T_8})$;

\STATE {\bf if} $C_4$ holds, {\bf then} let $S \leftarrow F(\setof{Z_1^+, T_5, T_7})$;

\STATE {\bf if} $C_6$ holds, {\bf then} let $S \leftarrow F(\setof{T_1, T_2, T_3})$;

\RETURN $S$.

\end{algorithmic}
\end{algorithm}
\baselineskip \skippt

\vspace{-10pt}
\begin{lemma}
\label{lemma:finding-Q_1}
Given the dominating sets $D(G_i, u_i)$, $P(G_i, u_i)$, $P'(G_i, u_i)$, and $\bar{P}(G_i, u_i)$ for $1 \le i \le k$, Procedure~\ref{algorithm:finding-Q_1} outputs a minimum-weight dominating set $S$ of $H$ such that $u_1 \not \in S$, $H[S]$ has a perfect matching, and $P'(G_1,u_1)\subseteq S$. Moreover, the procedure can be completed in $O(k)$ time.
\end{lemma}
\begin{proof}
By using a similar method of the arguments in Lemma~\ref{lemma:finding-D(H, u_1)}, one can show that all the paired-dominating sets $Z_1$, $Z_1^+$, $Z_1^-$, $T_1, \dots, T_8$ can be constructed in $O(k)$ time. Hence, the procedure certainly runs in $O(k)$ time. Further, one can verify that all the possible combinations of  conditions $D_1, D_2, \ldots, D_5$ have been considered in cases $C_1, C_2, \ldots, C_8$. Hence, to prove the correctness of the procedure, it suffices to show that each step of the procedure is correct.

First, we consider cases $C_1$ and $C_2$. In both of these cases, there exists an index $\ell$ such that $S_\ell^* = P(G_\ell,u_\ell)$. Therefore, $Z_1, Z_1^+$, and $Z_1^-$ are dominating sets of $H$. It follows that, if $r$ is even, then $Z_1$ is a minimum-weight dominating set of $H$ such that  $u_1 \not \in Z_1$ and $H[Z_1]$ has a perfect matching due to the selections of $S_i^*$ for $2 \le i \le k$. So, we have $S = Z_1$ for case $C_2$. On the other hand, if $r$ is odd, then in order to satisfy the condition that $H[Q_1]$ contains a perfect matching with minimum cost, we can either replace $S_\alpha^*$ with $D(G_\alpha,u_\alpha)$, or replace $S_\beta^*$ with $F(\setof{P(G_\beta, u_\beta), P'(G_\beta, u_\beta), \bar P(G_\beta, u_\beta)})$. This implies that we have $S = F(\setof{Z_1^+, Z_1^-})$ for case $C_1$. 

Next, we consider cases $C_3, C_4$, and $C_5$. Notice that in all three cases, there exists no index $\ell$ such that $S_\ell^* = P(G_\ell,u_\ell)$ and $r$ is an odd number. Moreover, for any paired-dominating set $Q_1$ of $H$, we have either $S_i^* = P'(G_i,u_i)$ for $2 \le i \le k$ or $V(B) \cap V(Q_1) \not = \emptyset$, where $B = H[\setof{u_1,u_2,\ldots,u_k}]$. In case $C_3$, a paired-dominating set $T_8$ is created for the former. Meanwhile, paired-dominating sets $Z_1^+$, $T_4$, and $T_6$ are built for the latter. As a consequence of $r = 1$, in order to ensure $H[Q_1]$ contains a perfect matching when $V(B) \cap V(Q_1) \not = \emptyset$, we replace $S_\alpha^*$ with $D(G_\alpha,u_\alpha)$ in $Z_1^+$, replace $S_\beta^*$ with $P(G_\beta,u_\beta)$ in $T_4$, and replace $S_\gamma^*$ and $S_\beta^*$ with $P(G_\gamma,u_\gamma)$ and $F(\{P'(G_\beta,u_\beta), \bar P(G_\beta,u_\beta)\})$ in $T_6$, respectively. Under the premise of minimizing weight, one can verify that $Z_1^+$, $T_4$, and $T_6$ are exactly the three potential candidates for $Q_1$. In case $C_4$, we have $S_\beta^* = D(G_\beta,u_\beta)$ and $S_i^* = P'(G_i,u_i)$ for $2 \le i \le k$ and $i \not = \beta$. Using a similar method of the above arguments, one can show that $S = F(\setof{Z_1^+, T_5, T_7})$ is true for case $C_4$. In case $C_5$, we have $r \ge 3$. Therefore, for the same reasons as case $C_1$, we have $S = F(\setof{Z_1^+, Z_1^-})$ for case $C_5$. 

Finally, we consider cases $C_6, C_7$, and $C_8$. Notice that in all three cases, there exists no index $\ell$ such that $S_\ell^* = P(G_\ell,u_\ell)$ and $r$ is an even number. In case $C_6$, we have either $S_i^* = P'(G_i,u_i)$ or $S_i^* = \bar{P}(G_i,u_i)$, where $1 \le i \le k$. To ensure $H[Q_1]$ contains a perfect matching, we replace $S_\gamma^*$ with $P(G_\gamma,u_\gamma)$ in $T_1$, replace $S_\alpha^*$ and $S_{\alpha'}^*$ with $D(G_\alpha,u_\alpha)$ and $D(G_{\alpha'},u_{\alpha'})$ in $T_2$, and replace $S_i^*$ with $P'(G_i,u_i))$ for all $i \in I$ in $T_3$, respectively. Under the premise of minimizing the weight $w(S)$, one can verify that $T_1$, $T_2$, and $T_3$ are exactly the three potential candidates for $Q_1$. Notice that, in case $C_7$, $S_i^* = P'(G_i,u_i)$ for $2 \le i \le k$. Further, $r \ge 2$ is an even number in case $C_8$. Thus, in both of these cases, we have $S = Z_1$ for the same reasons as case $C_2$.
\end{proof}
\vspace{3pt}
\subsubsection{Finding $Q_2$}\label{subsubsection:finding-Q_2}

In the following, we present a procedure to find the paired-dominating set $Q_2$. Similar to Procedure~\ref{algorithm:finding-Q_1}, the procedure solves the problem by considering six cases $C_9, C_{10}, \ldots, C_{14}$. For $9 \le i \le 14$, the case $C_i = (c_1, c_2,c_3,c_4)$ is an ordered 4-tuple. Further, the value of $c_k$ has the same definition as before for $1 \le k \le 4$. Then, we define $C_9 = (1,1,*,*)$, $C_{10} = (1,0,*,*)$, $C_{11} = (0,1,1,*)$, $C_{12} = (0,1,0,*)$, $C_{13} = (0,0,*,1)$, and $C_{14} = (0,0,*,0)$. Again, one can verify that all the possible combinations of the four conditions have been considered in cases $C_9, C_{10}, \ldots, C_{14}$. The paired-dominating sets $Z_2$, $Z_2^+$, $Z_2^-$, $T_9, \dots, T_{12}$ of $H$ are defined below, which are the potential candidates for $Q_2$.
\vspace{8pt}\\
\mbox{} \hspace{55pt} $Z_2 \hspace{3.4pt}= \bar{P}(G_1,u_1)\cup S_2^* \cup S_3^* \cup \ldots \cup S_k^*$. \\
\mbox{} \hspace{55pt} $Z_2^+ \hspace{0.2
pt}= (Z_1-S_\alpha^*) \cup D(G_\alpha,u_\alpha).$\\
\mbox{} \hspace{55pt} \hspace{0.1pt}$Z_2^- = (Z_1-S_\beta^*)\cup F(\{P(G_\beta, u_\beta),P'(G_\beta, u_\beta),\bar P(G_\beta, u_\beta)\})$.\\
\mbox{} \hspace{55pt} $T_9 \hspace{4.5pt} = (Z_2-S_\gamma^*) \cup P(G_\gamma, u_\gamma)$.\\
\mbox{} \hspace{55pt} $T_{10} \hspace{0.2pt} = (Z_2-S_\alpha^*-S_{\alpha'}^*) \cup D(G_\alpha,u_\alpha)\cup D(G_{\alpha'}, u_{\alpha'})$.\\
\mbox{} \hspace{55pt} $T_{11} \hspace{0.2pt} = (Z_2-S_\beta^*)\cup P(G_\beta,u_\beta)$.\\
\mbox{} \hspace{55pt} $T_{12} \hspace{0.2pt} =  (Z_2-S_\gamma^*-S_\beta^*)\cup P(G_\gamma,u_\gamma)\cup F(\{P'(G_\beta,u_\beta), \bar P(G_\beta,u_\beta)\})$.\\
\vspace{-16pt}

\noindent Moreover, the relations between the cases $C_9, C_{10}, \ldots, C_{14}$ and the paired-dominating sets $Z_2$, $Z_2^+$, $Z_2^-$, $T_9, \dots, T_{12}$ are detailed in Procedure~\ref{algorithm:finding-Q_2}.
\vspace{12pt}
\floatname{algorithm}{Procedure}
\begin{algorithm}
\caption{Finding $Q_2$}
\begin{algorithmic} [1]
\label{algorithm:finding-Q_2}
\medskip
\baselineskip 19pt
\REQUIRE A weighted block graph $H$ and a block $B$ of $H$ with 
$V(B) = \setof{u_1, u_2,\ldots, u_k}$. 
\\ \hspace{18pt} Dominating sets $D(G_i, u_i)$, $P(G_i, u_i)$, $P'(G_i, u_i)$, and $\bar{P}(G_i, u_i)$ for $1 \le i \le k$. 

\ENSURE A minimum-weight dominating set $S$ of $H$ such that $u_1 \not \in S$, $H[S]$ has a perfect \\ \hspace{33pt}matching, and $\bar{P}(G_1,u_1)\subseteq S$.

\STATE determine the paired-dominating sets $Z_2$, $Z_2^+$, $Z_2^-$, $T_9, \dots, T_{12}$;

\STATE {\bf if} $C_9$ or $C_{12}$ holds, {\bf then} let $S \leftarrow F(\setof{Z_2^+, Z_2^-})$;

\STATE {\bf if} $C_{10}$ or $C_{14}$ holds, {\bf then} let $S \leftarrow Z_2$;

\STATE {\bf if} $C_{11}$ holds, {\bf then} let $S \leftarrow F(\setof{Z_2^+, T_{11}, T_{12}})$;

\STATE {\bf if} $C_{13}$ holds, {\bf then} let $S \leftarrow F(\setof{T_9, T_{10}})$;

\RETURN $S$.

\end{algorithmic}
\end{algorithm}
\baselineskip \skippt

\vspace{-10pt}
\begin{lemma}
\label{lemma:finding-Q_2}
Given the dominating sets $D(G_i, u_i)$, $P(G_i, u_i)$, $P'(G_i, u_i)$, and $\bar{P}(G_i, u_i)$ for $1 \le i \le k$, Procedure~\ref{algorithm:finding-Q_2} outputs a minimum-weight dominating set $S$ of $H$ such that $u_1 \not \in S$, $H[S]$ has a perfect matching, and $\bar{P}(G_1,u_1)\subseteq S$. Moreover, the procedure can be completed in $O(k)$ time.
\end{lemma}
\vspace{-5pt}
\begin{proof}
By using a similar method of the arguments in Lemma~\ref{lemma:finding-D(H, u_1)}, one can show that each step of the procedure can be completed in $O(k)$ time. Therefore, the procedure runs in $O(k)$ time. Further, one can verify that all the possible combinations of conditions $D_1, D_2, D_3$, and $D_4$ have been considered in cases $C_9, C_{10}, \ldots, C_{14}$. Hence, to prove the correctness of the procedure, it suffices to show that each step of the procedure is correct.

First, we consider cases $C_9$ and $C_{10}$. In both of these cases, there exists an index $\ell$ such that $S_\ell^* = P(G_\ell,u_\ell)$. Therefore, for the same reasons as cases $C_1$ and $C_2$ in Procedure~$4$, we have $S = F(\setof{Z_2^+, Z_2^-})$ for case $C_9$ and $S = Z_2$ for case $C_{10}$, respectively. Next, we consider cases $C_{11}$, and $C_{12}$. Notice that in both cases, there exists no index $\ell$ such that $S_\ell^* = P(G_\ell,u_\ell)$ and $r$ is an odd number. Since we have $r = 1$ in case $C_{11}$, in order to satisfy the condition that $H[Q_2]$ contains a perfect matching with minimum cost, we replace $S_\alpha^*$ with $D(G_\alpha,u_\alpha)$ in $Z_2^+$, replace $S_\beta^*$ with $P(G_\beta,u_\beta)$ in $T_{11}$, and replace $S_\gamma^*$ and $S_\beta^*$ with $P(G_\gamma,u_\gamma)$ and $F(\{P'(G_\beta,u_\beta), \bar P(G_\beta,u_\beta)\})$ in $T_{12}$, respectively. Under the premise of minimizing weight, one can verify that $Z_2^+$, $T_{11}$, and $T_{12}$ are exactly the three potential candidates for $Q_2$. In case $C_{12}$, we have $r \ge 3$. Therefore, for the same reasons as case $C_9$, we have $S = F(\setof{Z_2^+, Z_2^-})$ for case $C_{12}$.

Finally, we consider cases $C_{13}$, and $C_{14}$. Notice that in both cases, there exists no index $\ell$ such that $S_\ell^* = P(G_\ell,u_\ell)$ and $r$ is an even number. In case $C_{13}$, either $S_i^* = P'(G_i,u_i)$ or $S_i^* = \bar{P}(G_i,u_i)$ for $1 \le i \le k$. Therefore, to satisfy the condition that $H[Q_2]$ contains a perfect matching, we replace $S_\gamma^*$ with $P(G_\gamma,u_\gamma)$ in $T_9$, and replace $S_\alpha^*$ and $S_{\alpha'}^*$ with $D(G_\alpha,u_\alpha)$ and $D(G_{\alpha'},u_{\alpha'})$ in $T_{10}$, respectively. Again, under the premise of minimizing the weight, one can verify that $T_9$ and $T_{10}$ are exactly the two potential candidates for $Q_2$. Notice that, in case $C_{14}$, $r \ge 2$ is an even number. Thus, we have $S = Z_2$ for the same reasons as case $C_{2}$ in Procedure~$4$. 
\end{proof}

\vspace{15pt}
\noindent Combining Lemmas~\ref{lemma:finding-Q_1} and~\ref{lemma:finding-Q_2}, we obtain the following result. 

\begin{lemma}
\label{lemma:finding-P'(H, u_1)}
Given the dominating sets $D(G_i, u_i)$, $P(G_i, u_i)$, $P'(G_i, u_i)$, and $\bar{P}(G_i, u_i)$ for $1 \le i \le k$, a $\kappa_3$-paired-dominating set $P'(H, u_1)$ can be determined in $O(k)$ time.
\end{lemma}

\subsection{Determination of $\bar{P}(H, u_1)$}\label{subsection:finding-bar_P(H, u_1)}

Remember that $\bar{P}(H,u_1)$ is a minimum-weight dominating set of $H-u_1$ and $u_1$ is not dominated by $\bar{P}(H, u_1)$. Hence, by the definition of $\bar{P}(H,u_1)$, the only composition is $\bar{P}(H, u_1) =  \bar{P}(G_1,u_1)\cup P'(G_2,u_2)\cup \ldots \cup P'(G_k,u_k)$. This implies that, given the dominating sets $D(G_i, u_i)$, $P(G_i, u_i)$, $P'(G_i, u_i)$, and $\bar{P}(G_i, u_i)$ for $1 \le i \le k$, a $\kappa_4$-paired-dominating set $\bar{P}(H, u_1)$ can be determined in $O(k)$ time. Thus, we have the following result. 

\begin{lemma}
\label{lemma:finding-bar_P(H, u_2)}
Given the dominating sets $D(G_i, u_i)$, $P(G_i, u_i)$, $P'(G_i, u_i)$, and $\bar{P}(G_i, u_i)$ for $1 \le i \le k$, a $\kappa_4$-paired-dominating set  $\bar{P}(H, u_1)$ can be determined in $O(k)$ time.
\end{lemma}

\section{Conclusion and Future Work}\label{section:conclusion}
In this paper, we have presented an optimal algorithm for finding a  paired-dominating set of a weighted block graph $G$. The algorithm uses dynamic programming to iteratively determine $D(H, u)$, $P(H, u)$, $P'(H, u)$, and $\bar{P}(H, u)$ in a bottom-up manner, where $H$ is a subgraph of $G$ and $u \in V(H)$ is a cut vertex of $G$. When the graph is given in an adjacency list representation, our algorithm runs in $O(n+m)$ time. Moreover, the algorithm can be completed in $O(n)$ time if the block-cut-vertex structure of $G$ is given.

Below we present some open problems related to the paired-domination 
problem. It is known that distance-hereditary graphs is a proper superfamily of block graphs. Therefore, it is interesting to study the time complexity of paired-domination problem in distance-hereditary graphs. In~\cite{Chen09}, Chen~{\em et al.} proposed an approximation algorithm with ratio $\ln(2\Delta(G))+1$ for general graphs and showed that the problem is APX-complete, i.e., has no PTAS. Thus, it would be useful if we could develop an approximation algorithm for general graphs with constant ratio. Meanwhile, it would be desirable to show that the problem remains NP-complete in planar graphs and design an approximation algorithm.

\small
\bibliographystyle{abbrv}
\bibliography{PDom}
\end{document}

%% file: block_graph.tex
\begin{picture}(0,0)%
\includegraphics{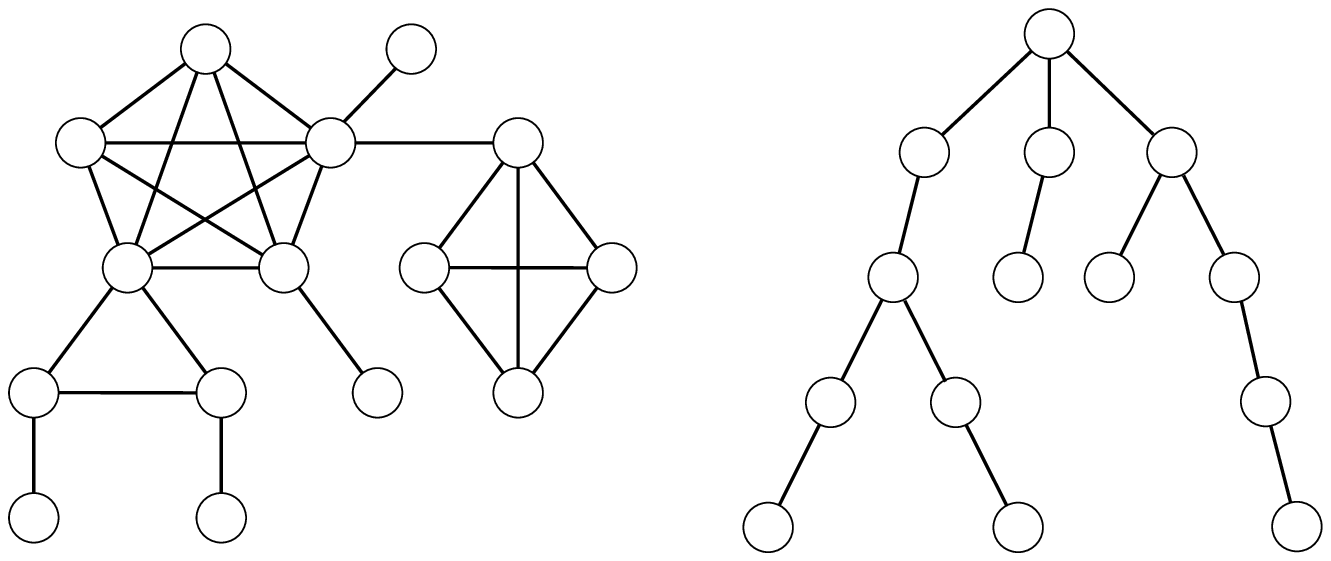}%
\end{picture}%
\setlength{\unitlength}{3947sp}%
\begingroup\makeatletter\ifx\SetFigFont\undefined%
\gdef\SetFigFont#1#2#3#4#5{%
  \reset@font\fontsize{#1}{#2pt}%
  \fontfamily{#3}\fontseries{#4}\fontshape{#5}%
  \selectfont}%
\fi\endgroup%
\begin{picture}(6618,3357)(436,-2371)
\put(1876,-886){\makebox(0,0)[lb]{\smash{{\SetFigFont{10}{12.0}{\rmdefault}{\mddefault}{\updefault}{\color[rgb]{0,0,0}$c_2$}%
}}}}
\put(3286,335){\makebox(0,0)[lb]{\smash{{\SetFigFont{10}{12.0}{\rmdefault}{\mddefault}{\updefault}{\color[rgb]{0,0,0}$c_5$}%
}}}}
\put(2404,335){\makebox(0,0)[lb]{\smash{{\SetFigFont{10}{12.0}{\rmdefault}{\mddefault}{\updefault}{\color[rgb]{0,0,0}$c_6$}%
}}}}
\put(451,-886){\makebox(0,0)[lb]{\smash{{\SetFigFont{10}{12.0}{\rmdefault}{\mddefault}{\updefault}{\color[rgb]{0,0,0}$c_1$}%
}}}}
\put(2176,-286){\makebox(0,0)[lb]{\smash{{\SetFigFont{10}{12.0}{\rmdefault}{\mddefault}{\updefault}{\color[rgb]{0,0,0}$c_4$}%
}}}}
\put(901,-286){\makebox(0,0)[lb]{\smash{{\SetFigFont{10}{12.0}{\rmdefault}{\mddefault}{\updefault}{\color[rgb]{0,0,0}$c_3$}%
}}}}
\put(3014,-988){\makebox(0,0)[lb]{\smash{{\SetFigFont{10}{12.0}{\rmdefault}{\mddefault}{\updefault}{\color[rgb]{0,0,0}$v_{13}$}%
}}}}
\put(1171,-388){\makebox(0,0)[lb]{\smash{{\SetFigFont{10}{12.0}{\rmdefault}{\mddefault}{\updefault}{\color[rgb]{0,0,0}$v_6$}%
}}}}
\put(1921,-388){\makebox(0,0)[lb]{\smash{{\SetFigFont{10}{12.0}{\rmdefault}{\mddefault}{\updefault}{\color[rgb]{0,0,0}$v_7$}%
}}}}
\put(2596,-388){\makebox(0,0)[lb]{\smash{{\SetFigFont{10}{12.0}{\rmdefault}{\mddefault}{\updefault}{\color[rgb]{0,0,0}$v_8$}%
}}}}
\put(3496,-388){\makebox(0,0)[lb]{\smash{{\SetFigFont{10}{12.0}{\rmdefault}{\mddefault}{\updefault}{\color[rgb]{0,0,0}$v_9$}%
}}}}
\put(946,212){\makebox(0,0)[lb]{\smash{{\SetFigFont{10}{12.0}{\rmdefault}{\mddefault}{\updefault}{\color[rgb]{0,0,0}$v_3$}%
}}}}
\put(3046,212){\makebox(0,0)[lb]{\smash{{\SetFigFont{10}{12.0}{\rmdefault}{\mddefault}{\updefault}{\color[rgb]{0,0,0}$v_5$}%
}}}}
\put(862,-1307){\makebox(0,0)[lb]{\smash{{\SetFigFont{10}{12.0}{\rmdefault}{\mddefault}{\updefault}{\color[rgb]{0,0,0}$B_1$}%
}}}}
\put(1438,-1307){\makebox(0,0)[lb]{\smash{{\SetFigFont{10}{12.0}{\rmdefault}{\mddefault}{\updefault}{\color[rgb]{0,0,0}$B_2$}%
}}}}
\put(1589,-1588){\makebox(0,0)[lb]{\smash{{\SetFigFont{10}{12.0}{\rmdefault}{\mddefault}{\updefault}{\color[rgb]{0,0,0}$v_{15}$}%
}}}}
\put(689,-1588){\makebox(0,0)[lb]{\smash{{\SetFigFont{10}{12.0}{\rmdefault}{\mddefault}{\updefault}{\color[rgb]{0,0,0}$v_{14}$}%
}}}}
\put(1546,662){\makebox(0,0)[lb]{\smash{{\SetFigFont{10}{12.0}{\rmdefault}{\mddefault}{\updefault}{\color[rgb]{0,0,0}$v_1$}%
}}}}
\put(2627,336){\makebox(0,0)[lb]{\smash{{\SetFigFont{10}{12.0}{\rmdefault}{\mddefault}{\updefault}{\color[rgb]{0,0,0}$B_6$}%
}}}}
\put(2146,212){\makebox(0,0)[lb]{\smash{{\SetFigFont{10}{12.0}{\rmdefault}{\mddefault}{\updefault}{\color[rgb]{0,0,0}$v_4$}%
}}}}
\put(1589,-988){\makebox(0,0)[lb]{\smash{{\SetFigFont{10}{12.0}{\rmdefault}{\mddefault}{\updefault}{\color[rgb]{0,0,0}$v_{11}$}%
}}}}
\put(1126,-809){\makebox(0,0)[lb]{\smash{{\SetFigFont{10}{12.0}{\rmdefault}{\mddefault}{\updefault}{\color[rgb]{0,0,0}$B_3$}%
}}}}
\put(2533,662){\makebox(0,0)[lb]{\smash{{\SetFigFont{10}{12.0}{\rmdefault}{\mddefault}{\updefault}{\color[rgb]{0,0,0}$v_2$}%
}}}}
\put(2185,553){\makebox(0,0)[lb]{\smash{{\SetFigFont{10}{12.0}{\rmdefault}{\mddefault}{\updefault}{\color[rgb]{0,0,0}$B_7$}%
}}}}
\put(4315,-886){\makebox(0,0)[lb]{\smash{{\SetFigFont{10}{12.0}{\rmdefault}{\mddefault}{\updefault}{\color[rgb]{0,0,0}$c_1$}%
}}}}
\put(5377,-886){\makebox(0,0)[lb]{\smash{{\SetFigFont{10}{12.0}{\rmdefault}{\mddefault}{\updefault}{\color[rgb]{0,0,0}$c_2$}%
}}}}
\put(4537,-361){\makebox(0,0)[lb]{\smash{{\SetFigFont{10}{12.0}{\rmdefault}{\mddefault}{\updefault}{\color[rgb]{0,0,0}$B_3$}%
}}}}
\put(5148,-361){\makebox(0,0)[lb]{\smash{{\SetFigFont{10}{12.0}{\rmdefault}{\mddefault}{\updefault}{\color[rgb]{0,0,0}$B_4$}%
}}}}
\put(6739,-361){\makebox(0,0)[lb]{\smash{{\SetFigFont{10}{12.0}{\rmdefault}{\mddefault}{\updefault}{\color[rgb]{0,0,0}$B_6$}%
}}}}
\put(7039,-1557){\makebox(0,0)[lb]{\smash{{\SetFigFont{10}{12.0}{\rmdefault}{\mddefault}{\updefault}{\color[rgb]{0,0,0}$B_5$}%
}}}}
\put(6865,-882){\makebox(0,0)[lb]{\smash{{\SetFigFont{10}{12.0}{\rmdefault}{\mddefault}{\updefault}{\color[rgb]{0,0,0}$c_5$}%
}}}}
\put(6127,-361){\makebox(0,0)[lb]{\smash{{\SetFigFont{10}{12.0}{\rmdefault}{\mddefault}{\updefault}{\color[rgb]{0,0,0}$B_7$}%
}}}}
\put(5851,839){\makebox(0,0)[lb]{\smash{{\SetFigFont{10}{12.0}{\rmdefault}{\mddefault}{\updefault}{\color[rgb]{0,0,0}$B_8$}%
}}}}
\put(5701,-1561){\makebox(0,0)[lb]{\smash{{\SetFigFont{10}{12.0}{\rmdefault}{\mddefault}{\updefault}{\color[rgb]{0,0,0}$B_2$}%
}}}}
\put(3937,-1561){\makebox(0,0)[lb]{\smash{{\SetFigFont{10}{12.0}{\rmdefault}{\mddefault}{\updefault}{\color[rgb]{0,0,0}$B_1$}%
}}}}
\put(5551,-2311){\makebox(0,0)[lb]{\smash{{\SetFigFont{10}{12.0}{\rmdefault}{\mddefault}{\updefault}{\color[rgb]{0,0,0}$(b)$}%
}}}}
\put(6427,314){\makebox(0,0)[lb]{\smash{{\SetFigFont{10}{12.0}{\rmdefault}{\mddefault}{\updefault}{\color[rgb]{0,0,0}$c_6$}%
}}}}
\put(5365,314){\makebox(0,0)[lb]{\smash{{\SetFigFont{10}{12.0}{\rmdefault}{\mddefault}{\updefault}{\color[rgb]{0,0,0}$c_4$}%
}}}}
\put(4765,314){\makebox(0,0)[lb]{\smash{{\SetFigFont{10}{12.0}{\rmdefault}{\mddefault}{\updefault}{\color[rgb]{0,0,0}$c_3$}%
}}}}
\put(2026,-2311){\makebox(0,0)[lb]{\smash{{\SetFigFont{10}{12.0}{\rmdefault}{\mddefault}{\updefault}{\color[rgb]{0,0,0}$(a)$}%
}}}}
\put(3403,-809){\makebox(0,0)[lb]{\smash{{\SetFigFont{10}{12.0}{\rmdefault}{\mddefault}{\updefault}{\color[rgb]{0,0,0}$B_5$}%
}}}}
\put(2251,-623){\makebox(0,0)[lb]{\smash{{\SetFigFont{10}{12.0}{\rmdefault}{\mddefault}{\updefault}{\color[rgb]{0,0,0}$B_4$}%
}}}}
\put(1072,553){\makebox(0,0)[lb]{\smash{{\SetFigFont{10}{12.0}{\rmdefault}{\mddefault}{\updefault}{\color[rgb]{0,0,0}$B_8$}%
}}}}
\put(689,-988){\makebox(0,0)[lb]{\smash{{\SetFigFont{10}{12.0}{\rmdefault}{\mddefault}{\updefault}{\color[rgb]{0,0,0}$v_{10}$}%
}}}}
\put(2339,-988){\makebox(0,0)[lb]{\smash{{\SetFigFont{10}{12.0}{\rmdefault}{\mddefault}{\updefault}{\color[rgb]{0,0,0}$v_{12}$}%
}}}}
\end{picture}%

%% file: union.tex
\begin{picture}(0,0)%
\includegraphics{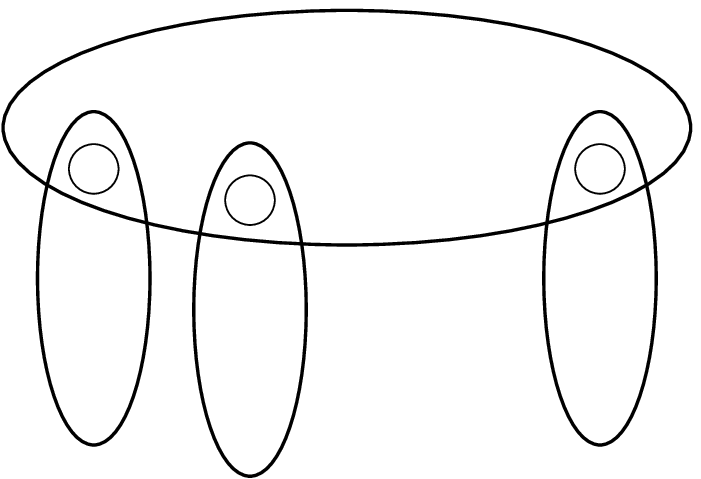}%
\end{picture}%
\setlength{\unitlength}{3947sp}%
\begingroup\makeatletter\ifx\SetFigFont\undefined%
\gdef\SetFigFont#1#2#3#4#5{%
  \reset@font\fontsize{#1}{#2pt}%
  \fontfamily{#3}\fontseries{#4}\fontshape{#5}%
  \selectfont}%
\fi\endgroup%
\begin{picture}(3330,2265)(-1664,-1410)
\put(226,-61){\makebox(0,0)[lb]{\smash{{\SetFigFont{12}{14.4}{\rmdefault}{\mddefault}{\updefault}{\color[rgb]{0,0,0}${\bf \ldots}$}%
}}}}
\put(251,-1111){\makebox(0,0)[lb]{\smash{{\SetFigFont{12}{14.4}{\rmdefault}{\mddefault}{\updefault}{\color[rgb]{0,0,0}${\bf \ldots}$}%
}}}}
\put(-54,314){\makebox(0,0)[lb]{\smash{{\SetFigFont{10}{12.0}{\rmdefault}{\mddefault}{\updefault}{\color[rgb]{0,0,0}$B$}%
}}}}
\put(-1289,-896){\makebox(0,0)[lb]{\smash{{\SetFigFont{10}{12.0}{\rmdefault}{\mddefault}{\updefault}{\color[rgb]{0,0,0}$G_1$}%
}}}}
\put(-539,-1046){\makebox(0,0)[lb]{\smash{{\SetFigFont{10}{12.0}{\rmdefault}{\mddefault}{\updefault}{\color[rgb]{0,0,0}$G_2$}%
}}}}
\put(1141,-896){\makebox(0,0)[lb]{\smash{{\SetFigFont{10}{12.0}{\rmdefault}{\mddefault}{\updefault}{\color[rgb]{0,0,0}$G_k$}%
}}}}
\put(-974,164){\makebox(0,0)[lb]{\smash{{\SetFigFont{10}{12.0}{\rmdefault}{\mddefault}{\updefault}{\color[rgb]{0,0,0}$u_1$}%
}}}}
\put(-224, 14){\makebox(0,0)[lb]{\smash{{\SetFigFont{10}{12.0}{\rmdefault}{\mddefault}{\updefault}{\color[rgb]{0,0,0}$u_2$}%
}}}}
\put(751,164){\makebox(0,0)[lb]{\smash{{\SetFigFont{10}{12.0}{\rmdefault}{\mddefault}{\updefault}{\color[rgb]{0,0,0}$u_k$}%
}}}}
\end{picture}%

%% file: example.tex
\begin{picture}(0,0)%
\includegraphics{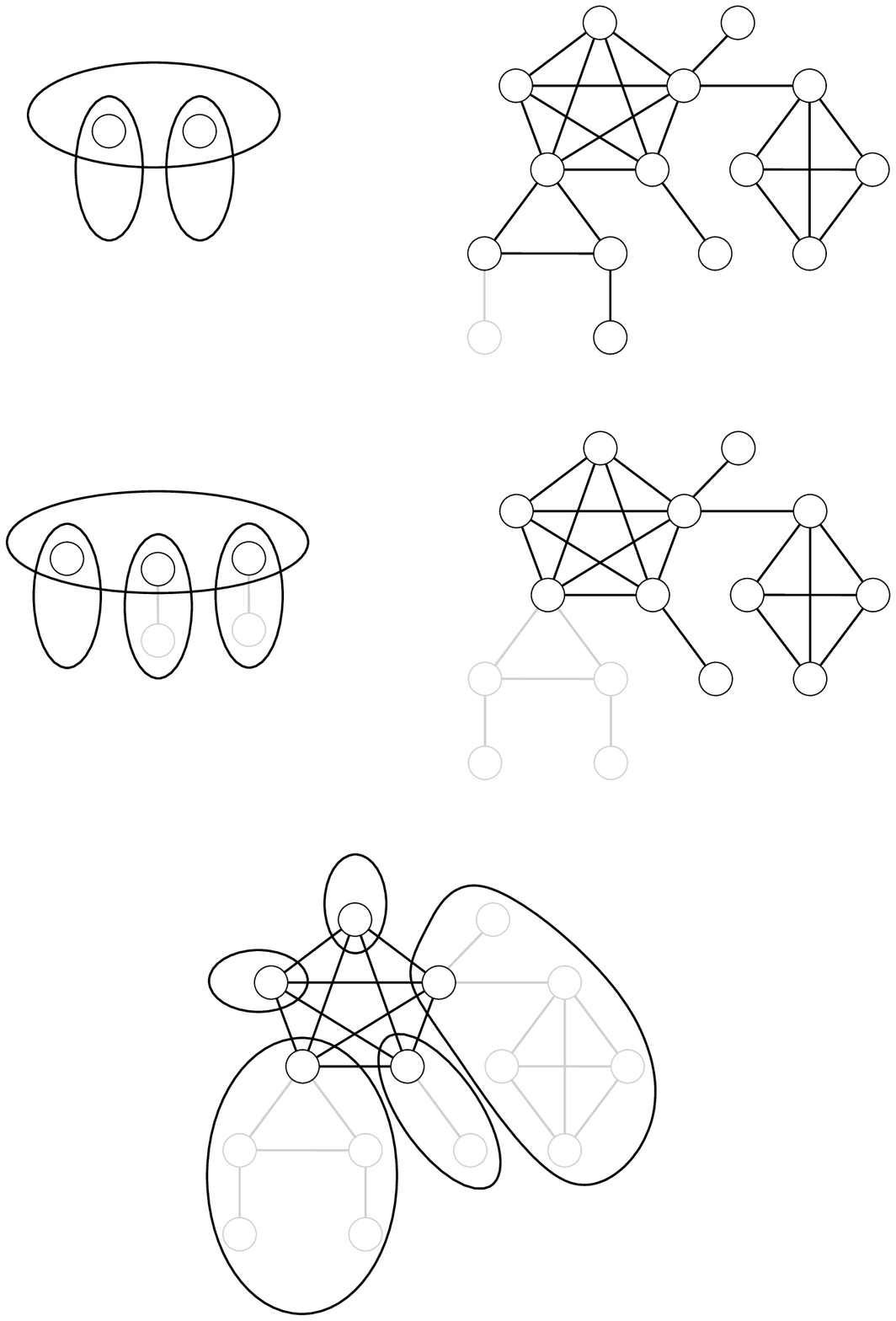}%
\end{picture}%
\setlength{\unitlength}{3947sp}%
\begingroup\makeatletter\ifx\SetFigFont\undefined%
\gdef\SetFigFont#1#2#3#4#5{%
  \reset@font\fontsize{#1}{#2pt}%
  \fontfamily{#3}\fontseries{#4}\fontshape{#5}%
  \selectfont}%
\fi\endgroup%
\begin{picture}(6409,9675)(-2742,-8920)
\put(1921,-2071){\makebox(0,0)[lb]{\smash{{\SetFigFont{10}{12.0}{\rmdefault}{\mddefault}{\updefault}{\color[rgb]{0,0,0}$(b)$}%
}}}}
\put(-1709, 69){\makebox(0,0)[lb]{\smash{{\SetFigFont{10}{12.0}{\rmdefault}{\mddefault}{\updefault}$B_1$}}}}
\put(-1754,-2071){\makebox(0,0)[lb]{\smash{{\SetFigFont{10}{12.0}{\rmdefault}{\mddefault}{\updefault}{\color[rgb]{0,0,0}$(a)$}%
}}}}
\put(-1382,-176){\makebox(0,0)[lb]{\smash{{\SetFigFont{10}{12.0}{\rmdefault}{\mddefault}{\updefault}{\color[rgb]{0,0,0}$v_{14}$}%
}}}}
\put(1921,-5131){\makebox(0,0)[lb]{\smash{{\SetFigFont{10}{12.0}{\rmdefault}{\mddefault}{\updefault}{\color[rgb]{0,0,0}$(d)$}%
}}}}
\put(-1754,-5131){\makebox(0,0)[lb]{\smash{{\SetFigFont{10}{12.0}{\rmdefault}{\mddefault}{\updefault}{\color[rgb]{0,0,0}$(c)$}%
}}}}
\put(-1724,-2935){\makebox(0,0)[lb]{\smash{{\SetFigFont{10}{12.0}{\rmdefault}{\mddefault}{\updefault}$B_3$}}}}
\put(-2299,-3233){\makebox(0,0)[lb]{\smash{{\SetFigFont{10}{12.0}{\rmdefault}{\mddefault}{\updefault}{\color[rgb]{0,0,0}$v_6$}%
}}}}
\put(-1030,-3231){\makebox(0,0)[lb]{\smash{{\SetFigFont{10}{12.0}{\rmdefault}{\mddefault}{\updefault}{\color[rgb]{0,0,0}$v_{11}$}%
}}}}
\put(-1680,-3818){\makebox(0,0)[lb]{\smash{{\SetFigFont{10}{12.0}{\rmdefault}{\mddefault}{\updefault}{\color[rgb]{0.816,0.816,0.816}$v_{14}$}%
}}}}
\put(-1680,-3308){\makebox(0,0)[lb]{\smash{{\SetFigFont{10}{12.0}{\rmdefault}{\mddefault}{\updefault}{\color[rgb]{0,0,0}$v_{10}$}%
}}}}
\put(-2727,-4242){\makebox(0,0)[lb]{\smash{{\SetFigFont{10}{12.0}{\rmdefault}{\mddefault}{\updefault}{\color[rgb]{0,0,0}$G[\{v_6\}]$}%
}}}}
\put(-2003,-4362){\makebox(0,0)[lb]{\smash{{\SetFigFont{10}{12.0}{\rmdefault}{\mddefault}{\updefault}{\color[rgb]{0,0,0}$G[\{v_{10},v_{14}\}]$}%
}}}}
\put(-954,-4242){\makebox(0,0)[lb]{\smash{{\SetFigFont{10}{12.0}{\rmdefault}{\mddefault}{\updefault}{\color[rgb]{0,0,0}$G[\{v_{11},v_{15}\}]$}%
}}}}
\put(166,-8860){\makebox(0,0)[lb]{\smash{{\SetFigFont{10}{12.0}{\rmdefault}{\mddefault}{\updefault}{\color[rgb]{0,0,0}$(e)$}%
}}}}
\put(-1577,-1179){\makebox(0,0)[lb]{\smash{{\SetFigFont{10}{12.0}{\rmdefault}{\mddefault}{\updefault}{\color[rgb]{0,0,0}$G[\{v_{14}\}]$}%
}}}}
\put(-1030,-3741){\makebox(0,0)[lb]{\smash{{\SetFigFont{10}{12.0}{\rmdefault}{\mddefault}{\updefault}{\color[rgb]{0.816,0.816,0.816}$v_{15}$}%
}}}}
\put(3369,-872){\makebox(0,0)[lb]{\smash{{\SetFigFont{10}{12.0}{\rmdefault}{\mddefault}{\updefault}{\color[rgb]{0,0,0}$B_5$}%
}}}}
\put(2217,-686){\makebox(0,0)[lb]{\smash{{\SetFigFont{10}{12.0}{\rmdefault}{\mddefault}{\updefault}{\color[rgb]{0,0,0}$B_4$}%
}}}}
\put(1038,490){\makebox(0,0)[lb]{\smash{{\SetFigFont{10}{12.0}{\rmdefault}{\mddefault}{\updefault}{\color[rgb]{0,0,0}$B_8$}%
}}}}
\put(655,-1051){\makebox(0,0)[lb]{\smash{{\SetFigFont{10}{12.0}{\rmdefault}{\mddefault}{\updefault}{\color[rgb]{0,0,0}$v_{10}$}%
}}}}
\put(2305,-1051){\makebox(0,0)[lb]{\smash{{\SetFigFont{10}{12.0}{\rmdefault}{\mddefault}{\updefault}{\color[rgb]{0,0,0}$v_{12}$}%
}}}}
\put(2980,-1051){\makebox(0,0)[lb]{\smash{{\SetFigFont{10}{12.0}{\rmdefault}{\mddefault}{\updefault}{\color[rgb]{0,0,0}$v_{13}$}%
}}}}
\put(1137,-451){\makebox(0,0)[lb]{\smash{{\SetFigFont{10}{12.0}{\rmdefault}{\mddefault}{\updefault}{\color[rgb]{0,0,0}$v_6$}%
}}}}
\put(1887,-451){\makebox(0,0)[lb]{\smash{{\SetFigFont{10}{12.0}{\rmdefault}{\mddefault}{\updefault}{\color[rgb]{0,0,0}$v_7$}%
}}}}
\put(2562,-451){\makebox(0,0)[lb]{\smash{{\SetFigFont{10}{12.0}{\rmdefault}{\mddefault}{\updefault}{\color[rgb]{0,0,0}$v_8$}%
}}}}
\put(3462,-451){\makebox(0,0)[lb]{\smash{{\SetFigFont{10}{12.0}{\rmdefault}{\mddefault}{\updefault}{\color[rgb]{0,0,0}$v_9$}%
}}}}
\put(912,149){\makebox(0,0)[lb]{\smash{{\SetFigFont{10}{12.0}{\rmdefault}{\mddefault}{\updefault}{\color[rgb]{0,0,0}$v_3$}%
}}}}
\put(3012,149){\makebox(0,0)[lb]{\smash{{\SetFigFont{10}{12.0}{\rmdefault}{\mddefault}{\updefault}{\color[rgb]{0,0,0}$v_5$}%
}}}}
\put(828,-1370){\makebox(0,0)[lb]{\smash{{\SetFigFont{10}{12.0}{\rmdefault}{\mddefault}{\updefault}{\color[rgb]{0.816,0.816,0.816}$B_1$}%
}}}}
\put(1404,-1370){\makebox(0,0)[lb]{\smash{{\SetFigFont{10}{12.0}{\rmdefault}{\mddefault}{\updefault}{\color[rgb]{0,0,0}$B_2$}%
}}}}
\put(1555,-1651){\makebox(0,0)[lb]{\smash{{\SetFigFont{10}{12.0}{\rmdefault}{\mddefault}{\updefault}{\color[rgb]{0,0,0}$v_{15}$}%
}}}}
\put(655,-1651){\makebox(0,0)[lb]{\smash{{\SetFigFont{10}{12.0}{\rmdefault}{\mddefault}{\updefault}{\color[rgb]{0.816,0.816,0.816}$v_{14}$}%
}}}}
\put(1512,599){\makebox(0,0)[lb]{\smash{{\SetFigFont{10}{12.0}{\rmdefault}{\mddefault}{\updefault}{\color[rgb]{0,0,0}$v_1$}%
}}}}
\put(2557,273){\makebox(0,0)[lb]{\smash{{\SetFigFont{10}{12.0}{\rmdefault}{\mddefault}{\updefault}{\color[rgb]{0,0,0}$B_6$}%
}}}}
\put(2151,490){\makebox(0,0)[lb]{\smash{{\SetFigFont{10}{12.0}{\rmdefault}{\mddefault}{\updefault}{\color[rgb]{0,0,0}$B_7$}%
}}}}
\put(2112,149){\makebox(0,0)[lb]{\smash{{\SetFigFont{10}{12.0}{\rmdefault}{\mddefault}{\updefault}{\color[rgb]{0,0,0}$v_4$}%
}}}}
\put(1555,-1051){\makebox(0,0)[lb]{\smash{{\SetFigFont{10}{12.0}{\rmdefault}{\mddefault}{\updefault}{\color[rgb]{0,0,0}$v_{11}$}%
}}}}
\put(1092,-872){\makebox(0,0)[lb]{\smash{{\SetFigFont{10}{12.0}{\rmdefault}{\mddefault}{\updefault}{\color[rgb]{0,0,0}$B_3$}%
}}}}
\put(2499,599){\makebox(0,0)[lb]{\smash{{\SetFigFont{10}{12.0}{\rmdefault}{\mddefault}{\updefault}{\color[rgb]{0,0,0}$v_2$}%
}}}}
\put(-2031,-176){\makebox(0,0)[lb]{\smash{{\SetFigFont{10}{12.0}{\rmdefault}{\mddefault}{\updefault}{\color[rgb]{0,0,0}$v_{10}$}%
}}}}
\put(-2291,-1179){\makebox(0,0)[lb]{\smash{{\SetFigFont{10}{12.0}{\rmdefault}{\mddefault}{\updefault}{\color[rgb]{0,0,0}$G[\{v_{10}\}]$}%
}}}}
\put(1618,-7285){\makebox(0,0)[lb]{\smash{{\SetFigFont{10}{12.0}{\rmdefault}{\mddefault}{\updefault}{\color[rgb]{0.816,0.816,0.816}$B_5$}%
}}}}
\put(466,-7099){\makebox(0,0)[lb]{\smash{{\SetFigFont{10}{12.0}{\rmdefault}{\mddefault}{\updefault}{\color[rgb]{0.816,0.816,0.816}$B_4$}%
}}}}
\put(-713,-5923){\makebox(0,0)[lb]{\smash{{\SetFigFont{10}{12.0}{\rmdefault}{\mddefault}{\updefault}{\color[rgb]{0,0,0}$B_8$}%
}}}}
\put(-1096,-7464){\makebox(0,0)[lb]{\smash{{\SetFigFont{10}{12.0}{\rmdefault}{\mddefault}{\updefault}{\color[rgb]{0.816,0.816,0.816}$v_{10}$}%
}}}}
\put(554,-7464){\makebox(0,0)[lb]{\smash{{\SetFigFont{10}{12.0}{\rmdefault}{\mddefault}{\updefault}{\color[rgb]{0.816,0.816,0.816}$v_{12}$}%
}}}}
\put(1229,-7464){\makebox(0,0)[lb]{\smash{{\SetFigFont{10}{12.0}{\rmdefault}{\mddefault}{\updefault}{\color[rgb]{0.816,0.816,0.816}$v_{13}$}%
}}}}
\put(-614,-6864){\makebox(0,0)[lb]{\smash{{\SetFigFont{10}{12.0}{\rmdefault}{\mddefault}{\updefault}{\color[rgb]{0,0,0}$v_6$}%
}}}}
\put(136,-6864){\makebox(0,0)[lb]{\smash{{\SetFigFont{10}{12.0}{\rmdefault}{\mddefault}{\updefault}{\color[rgb]{0,0,0}$v_7$}%
}}}}
\put(811,-6864){\makebox(0,0)[lb]{\smash{{\SetFigFont{10}{12.0}{\rmdefault}{\mddefault}{\updefault}{\color[rgb]{0.816,0.816,0.816}$v_8$}%
}}}}
\put(1711,-6864){\makebox(0,0)[lb]{\smash{{\SetFigFont{10}{12.0}{\rmdefault}{\mddefault}{\updefault}{\color[rgb]{0.816,0.816,0.816}$v_9$}%
}}}}
\put(-839,-6264){\makebox(0,0)[lb]{\smash{{\SetFigFont{10}{12.0}{\rmdefault}{\mddefault}{\updefault}{\color[rgb]{0,0,0}$v_3$}%
}}}}
\put(1261,-6264){\makebox(0,0)[lb]{\smash{{\SetFigFont{10}{12.0}{\rmdefault}{\mddefault}{\updefault}{\color[rgb]{0.816,0.816,0.816}$v_5$}%
}}}}
\put(-923,-7783){\makebox(0,0)[lb]{\smash{{\SetFigFont{10}{12.0}{\rmdefault}{\mddefault}{\updefault}{\color[rgb]{0.816,0.816,0.816}$B_1$}%
}}}}
\put(-347,-7783){\makebox(0,0)[lb]{\smash{{\SetFigFont{10}{12.0}{\rmdefault}{\mddefault}{\updefault}{\color[rgb]{0.816,0.816,0.816}$B_2$}%
}}}}
\put(-196,-8064){\makebox(0,0)[lb]{\smash{{\SetFigFont{10}{12.0}{\rmdefault}{\mddefault}{\updefault}{\color[rgb]{0.816,0.816,0.816}$v_{15}$}%
}}}}
\put(-1096,-8064){\makebox(0,0)[lb]{\smash{{\SetFigFont{10}{12.0}{\rmdefault}{\mddefault}{\updefault}{\color[rgb]{0.816,0.816,0.816}$v_{14}$}%
}}}}
\put(-239,-5814){\makebox(0,0)[lb]{\smash{{\SetFigFont{10}{12.0}{\rmdefault}{\mddefault}{\updefault}{\color[rgb]{0,0,0}$v_1$}%
}}}}
\put(806,-6140){\makebox(0,0)[lb]{\smash{{\SetFigFont{10}{12.0}{\rmdefault}{\mddefault}{\updefault}{\color[rgb]{0.816,0.816,0.816}$B_6$}%
}}}}
\put(400,-5923){\makebox(0,0)[lb]{\smash{{\SetFigFont{10}{12.0}{\rmdefault}{\mddefault}{\updefault}{\color[rgb]{0.816,0.816,0.816}$B_7$}%
}}}}
\put(361,-6264){\makebox(0,0)[lb]{\smash{{\SetFigFont{10}{12.0}{\rmdefault}{\mddefault}{\updefault}{\color[rgb]{0,0,0}$v_4$}%
}}}}
\put(-196,-7464){\makebox(0,0)[lb]{\smash{{\SetFigFont{10}{12.0}{\rmdefault}{\mddefault}{\updefault}{\color[rgb]{0.816,0.816,0.816}$v_{11}$}%
}}}}
\put(-659,-7285){\makebox(0,0)[lb]{\smash{{\SetFigFont{10}{12.0}{\rmdefault}{\mddefault}{\updefault}{\color[rgb]{0.816,0.816,0.816}$B_3$}%
}}}}
\put(748,-5814){\makebox(0,0)[lb]{\smash{{\SetFigFont{10}{12.0}{\rmdefault}{\mddefault}{\updefault}{\color[rgb]{0.816,0.816,0.816}$v_2$}%
}}}}
\put(3373,-3914){\makebox(0,0)[lb]{\smash{{\SetFigFont{10}{12.0}{\rmdefault}{\mddefault}{\updefault}{\color[rgb]{0,0,0}$B_5$}%
}}}}
\put(2221,-3728){\makebox(0,0)[lb]{\smash{{\SetFigFont{10}{12.0}{\rmdefault}{\mddefault}{\updefault}{\color[rgb]{0,0,0}$B_4$}%
}}}}
\put(1042,-2552){\makebox(0,0)[lb]{\smash{{\SetFigFont{10}{12.0}{\rmdefault}{\mddefault}{\updefault}{\color[rgb]{0,0,0}$B_8$}%
}}}}
\put(659,-4093){\makebox(0,0)[lb]{\smash{{\SetFigFont{10}{12.0}{\rmdefault}{\mddefault}{\updefault}{\color[rgb]{0.816,0.816,0.816}$v_{10}$}%
}}}}
\put(2309,-4093){\makebox(0,0)[lb]{\smash{{\SetFigFont{10}{12.0}{\rmdefault}{\mddefault}{\updefault}{\color[rgb]{0,0,0}$v_{12}$}%
}}}}
\put(2984,-4093){\makebox(0,0)[lb]{\smash{{\SetFigFont{10}{12.0}{\rmdefault}{\mddefault}{\updefault}{\color[rgb]{0,0,0}$v_{13}$}%
}}}}
\put(1141,-3493){\makebox(0,0)[lb]{\smash{{\SetFigFont{10}{12.0}{\rmdefault}{\mddefault}{\updefault}{\color[rgb]{0,0,0}$v_6$}%
}}}}
\put(1891,-3493){\makebox(0,0)[lb]{\smash{{\SetFigFont{10}{12.0}{\rmdefault}{\mddefault}{\updefault}{\color[rgb]{0,0,0}$v_7$}%
}}}}
\put(2566,-3493){\makebox(0,0)[lb]{\smash{{\SetFigFont{10}{12.0}{\rmdefault}{\mddefault}{\updefault}{\color[rgb]{0,0,0}$v_8$}%
}}}}
\put(3466,-3493){\makebox(0,0)[lb]{\smash{{\SetFigFont{10}{12.0}{\rmdefault}{\mddefault}{\updefault}{\color[rgb]{0,0,0}$v_9$}%
}}}}
\put(916,-2893){\makebox(0,0)[lb]{\smash{{\SetFigFont{10}{12.0}{\rmdefault}{\mddefault}{\updefault}{\color[rgb]{0,0,0}$v_3$}%
}}}}
\put(3016,-2893){\makebox(0,0)[lb]{\smash{{\SetFigFont{10}{12.0}{\rmdefault}{\mddefault}{\updefault}{\color[rgb]{0,0,0}$v_5$}%
}}}}
\put(832,-4412){\makebox(0,0)[lb]{\smash{{\SetFigFont{10}{12.0}{\rmdefault}{\mddefault}{\updefault}{\color[rgb]{0.816,0.816,0.816}$B_1$}%
}}}}
\put(1408,-4412){\makebox(0,0)[lb]{\smash{{\SetFigFont{10}{12.0}{\rmdefault}{\mddefault}{\updefault}{\color[rgb]{0.816,0.816,0.816}$B_2$}%
}}}}
\put(1559,-4693){\makebox(0,0)[lb]{\smash{{\SetFigFont{10}{12.0}{\rmdefault}{\mddefault}{\updefault}{\color[rgb]{0.816,0.816,0.816}$v_{15}$}%
}}}}
\put(659,-4693){\makebox(0,0)[lb]{\smash{{\SetFigFont{10}{12.0}{\rmdefault}{\mddefault}{\updefault}{\color[rgb]{0.816,0.816,0.816}$v_{14}$}%
}}}}
\put(1516,-2443){\makebox(0,0)[lb]{\smash{{\SetFigFont{10}{12.0}{\rmdefault}{\mddefault}{\updefault}{\color[rgb]{0,0,0}$v_1$}%
}}}}
\put(2561,-2769){\makebox(0,0)[lb]{\smash{{\SetFigFont{10}{12.0}{\rmdefault}{\mddefault}{\updefault}{\color[rgb]{0,0,0}$B_6$}%
}}}}
\put(2155,-2552){\makebox(0,0)[lb]{\smash{{\SetFigFont{10}{12.0}{\rmdefault}{\mddefault}{\updefault}{\color[rgb]{0,0,0}$B_7$}%
}}}}
\put(2116,-2893){\makebox(0,0)[lb]{\smash{{\SetFigFont{10}{12.0}{\rmdefault}{\mddefault}{\updefault}{\color[rgb]{0,0,0}$v_4$}%
}}}}
\put(1559,-4093){\makebox(0,0)[lb]{\smash{{\SetFigFont{10}{12.0}{\rmdefault}{\mddefault}{\updefault}{\color[rgb]{0.816,0.816,0.816}$v_{11}$}%
}}}}
\put(1096,-3914){\makebox(0,0)[lb]{\smash{{\SetFigFont{10}{12.0}{\rmdefault}{\mddefault}{\updefault}{\color[rgb]{0.816,0.816,0.816}$B_3$}%
}}}}
\put(2503,-2443){\makebox(0,0)[lb]{\smash{{\SetFigFont{10}{12.0}{\rmdefault}{\mddefault}{\updefault}{\color[rgb]{0,0,0}$v_2$}%
}}}}
\end{picture}%

%% file: Weighted_Paired-Domination_Problem_on_Block_Graphs-2016-02-20.bbl
\begin{thebibliography}{10}

\bibitem{Aho74}
A.~V. Aho, J.~E. Hopcroft, and J.~D. Ullman.
\newblock {\em The design and analysis of computer algorithms}.
\newblock Addison-Wesley Publishing Co., Reading, Mass.-London-Amsterdam, 1974.

\bibitem{Chang13}
G.~J. Chang.
\newblock Algorithmic aspects of domination in graphs.
\newblock In {\em Handbook of Combinatorial Optimization}, pages 339--405.
  Springer-Verlag, New York, second edition, 2013.

\bibitem{Chen09}
L.~Chen, C.~Lu, and Z.~Zeng.
\newblock Hardness results and approximation algorithms for (weighted)
  paired-domination graphs.
\newblock {\em Theoret. Comput. Sci.}, 410(47-49):5063--5071, 2009.

\bibitem{Chen-Lu-Zeng-09}
L.~Chen, C.~Lu, and Z.~Zeng.
\newblock A linear-time algorithm for paired-domination problem in strongly
  chordal graphs.
\newblock {\em Inform. Process. Lett.}, 110(1):20--23, 2009.

\bibitem{Chen10}
L.~Chen, C.~Lu, and Z.~Zeng.
\newblock Labelling algorithms for paired-domination problems in block and
  interval graphs.
\newblock {\em J. Comb. Optim.}, 19(4):457--470, 2010.

\bibitem{CKS09}
T.~C.~E. Cheng, L.~Kang, and E.~Shan.
\newblock A polynomial-time algorithm for the paired-domination problem on
  permutation graphs.
\newblock {\em Discrete Appl. Math.}, 157(2):262--271, 2009.

\bibitem{Goddard13}
W.~Goddard and M.~A. Henning.
\newblock Independent domination in graphs: a survey and recent results.
\newblock {\em Discrete Math.}, 313(7):839--854, 2013.

\bibitem{HS98}
T.~Haynes and P.~Slater.
\newblock Paired-domination in graphs.
\newblock {\em Networks}, 32:199--206, 1998.

\bibitem{Haynes98-2}
T.~W. Haynes, S.~T. Hedetniemi, and P.~J. Slater.
\newblock {\em Domination in Graphs: Advanced Topics}.
\newblock Marcel Dekker Inc., New York, 1998.

\bibitem{Haynes98}
T.~W. Haynes, S.~T. Hedetniemi, and P.~J. Slater.
\newblock {\em Fundamentals of Domination in Graphs}.
\newblock Marcel Dekker Inc., New York, 1998.

\bibitem{Hedetniemi90}
S.~T. Hedetniemi and R.~C. Laskar.
\newblock Bibliography on domination in graphs and some basic definitions of
  domination parameters.
\newblock {\em Discrete Math.}, 86(1-3):257--277, 1990.

\bibitem{Hedetniemi91}
S.~T. Hedetniemi and R.~C. Laskar, editors.
\newblock {\em Topics on domination}.
\newblock Annals of Discrete Mathematics. North-Holland Publishing Co.,
  Amsterdam, 1991.

\bibitem{Henning09}
M.~A. Henning.
\newblock A survey of selected recent results on total domination in graphs.
\newblock {\em Discrete Math.}, 309(1):32--63, 2009.

\bibitem{Hung12}
R.-W. Hung.
\newblock Linear-time algorithm for the paired-domination problem in convex
  bipartite graphs.
\newblock {\em Theory Comput. Syst.}, 50(4):721--738, 2012.

\bibitem{Kang13}
L.~Kang.
\newblock Variations of dominating set problem.
\newblock In {\em Handbook of Combinatorial Optimization}, pages 3363--3394.
  Springer-Verlag, New York, second edition, 2013.

\bibitem{KSC04}
L.~Kang, M.~Y. Sohn, and T.~C.~E. Cheng.
\newblock Paired-domination in inflated graphs.
\newblock {\em Theor. Comput. Sci.}, 320(2-3):485--494, 2004.

\bibitem{Lappas09}
E.~Lappas, S.~D. Nikolopoulos, and L.~Palios.
\newblock An {$O(n)$}-time algorithm for the paired-domination problem on
  permutation graphs.
\newblock In {\em Combinatorial algorithms}, volume 5874 of {\em Lecture Notes
  in Comput. Sci.}, pages 368--379. Springer, Berlin, 2009.

\bibitem{Lappas13}
E.~Lappas, S.~D. Nikolopoulos, and L.~Palios.
\newblock An {$O(n)$}-time algorithm for the paired domination problem on
  permutation graphs.
\newblock {\em European J. Combin.}, 34(3):593--608, 2013.

\bibitem{lin15}
C.-C. Lin and H.-L. Tu.
\newblock A linear-time algorithm for paired-domination on circular-arc graphs.
\newblock {\em Theoret. Comput. Sci.}, 591:99--105, 2015.

\bibitem{Panda13}
B.~S. Panda and D.~Pradhan.
\newblock A linear time algorithm for computing a minimum paired-dominating set
  of a convex bipartite graph.
\newblock {\em Discrete Appl. Math.}, 161(12):1776--1783, 2013.

\bibitem{Panda12}
B.~S. Panda and D.~Pradhan.
\newblock Minimum paired-dominating set in chordal bipartite graphs and perfect
  elimination bipartite graphs.
\newblock {\em J. Comb. Optim.}, 26(4):770--785, 2013.

\bibitem{QKCD03}
H.~Qiao, L.~Kang, M.~Cardei, and D.-Z. Du.
\newblock Paired-domination of trees.
\newblock {\em J. of Global Optimization}, 25(1):43--54, 2003.

\bibitem{Yen96}
C.-C. Yen and R.~C.~T. Lee.
\newblock The weighted perfect domination problem and its variants.
\newblock {\em Discrete Appl. Math.}, 66(2):147--160, 1996.

\end{thebibliography}
